\documentclass{amsart}
\usepackage{amsmath,amssymb}
\usepackage{mathtools}
\usepackage{bbm}
\usepackage[all,cmtip]{xy}
\usepackage{stmaryrd} %

\usepackage[colorinlistoftodos]{todonotes}

\usepackage{graphicx}
\graphicspath{ {big_draft/pictures/} {paper1_HarmEmb/pictures} {paper2_RecoThm/pictures} {paper3_GeneralAnalysis/pictures} }

\usepackage{fonttable}

\usepackage{datetime}

\usepackage{bm}

\usepackage{wrapfig}

\makeindex \topmargin=0in \headheight=8pt \headsep=0.5in \textheight=8.5in %
\usepackage{etoolbox}
\patchcmd{\section}{\scshape}{\large\bfseries}{}{}
\makeatletter
\renewcommand{\@secnumfont}{\bfseries}
\makeatother
    
\usepackage[style=alphabetic, backend=bibtex]{biblatex}
\usepackage{booktabs}
\bibliography{references}

\numberwithin{equation}{section}

\newtheorem{lemma}{Lemma}[section]
\newtheorem{thma}{Theorem}
\newtheorem{thmas}[lemma]{Theorem}

\newtheorem{cor}[lemma]{Corollary}

\theoremstyle{remark}
\newtheorem{rem}[lemma]{Remark}
\newtheorem{defin}[lemma]{Definition}

\usepackage{xcolor}

\usepackage{eucal}
\usepackage{enumitem}

\usepackage{hyperref}
\hypersetup{
    colorlinks=true,
    linkcolor=blue,
    filecolor=magenta,      
    urlcolor=cyan,
    citecolor=blue,
}

\newcommand{\nc}{\newcommand}

\nc{\Aa}{{\CMcal{A}}}
\nc{\Bb}{{\CMcal{B}}}
\nc{\Cc}{{\CMcal{C}}}
\nc{\Dd}{{\CMcal{D}}}
\nc{\Ee}{{\CMcal{E}}}
\nc{\Ff}{{\CMcal{F}}}
\nc{\Gg}{{\CMcal{G}}}
\nc{\Hh}{{\CMcal{H}}}
\nc{\Ii}{{\CMcal{I}}}
\nc{\Jj}{{\CMcal{J}}}
\nc{\Kk}{{\CMcal{K}}}
\nc{\Ll}{{\CMcal{L}}}
\nc{\Mm}{{\CMcal{M}}}
\nc{\Nn}{{\CMcal{N}}}
\nc{\Oo}{{\CMcal{O}}}
\nc{\Pp}{{\CMcal{P}}}
\nc{\Qq}{{\CMcal{Q}}}
\nc{\Rr}{{\CMcal{R}}}
\nc{\Ss}{{\CMcal{S}}}
\nc{\Tt}{{\CMcal{T}}}
\nc{\Uu}{{\CMcal{U}}}
\nc{\Vv}{{\CMcal{V}}}
\nc{\Ww}{{\CMcal{W}}}
\nc{\Xx}{{\CMcal{X}}}
\nc{\Yy}{{\CMcal{Y}}}
\nc{\Zz}{{\CMcal{Z}}}

\nc{\mA}{{\mathrm{A}}}
\nc{\mB}{{\mathrm{B}}}
\nc{\mC}{{\mathrm{C}}}
\nc{\mD}{{\mathrm{D}}}
\nc{\mE}{{\mathrm{E}}}
\nc{\mF}{{\mathrm{F}}}
\nc{\mG}{{\mathrm{G}}}
\nc{\mH}{{\mathrm{H}}}
\nc{\mI}{{\mathrm{I}}}
\nc{\mJ}{{\mathrm{J}}}
\nc{\mK}{{\mathrm{K}}}
\nc{\mL}{{\mathrm{L}}}
\nc{\mM}{{\mathrm{M}}}
\nc{\mN}{{\mathrm{N}}}
\nc{\mO}{{\mathrm{O}}}
\nc{\mP}{{\mathrm{P}}}
\nc{\mQ}{{\mathrm{Q}}}
\nc{\mR}{{\mathrm{R}}}
\nc{\mS}{{\mathrm{S}}}
\nc{\mT}{{\mathrm{T}}}
\nc{\mU}{{\mathrm{U}}}
\nc{\mV}{{\mathrm{V}}}
\nc{\mW}{{\mathrm{W}}}
\nc{\mX}{{\mathrm{X}}}
\nc{\mY}{{\mathrm{Y}}}
\nc{\mZ}{{\mathrm{Z}}}

\nc{\BB}{{\mathbb{B}}}
\nc{\CC}{{\mathbb{C}}}
\nc{\DD}{{\mathbb{D}}}
\DeclareMathOperator{\EE}{{\mathbb{E}}}
\nc{\FF}{{\mathbb{F}}}
\nc{\GG}{{\mathbb{G}}}
\nc{\HH}{{\mathbb{H}}}
\nc{\II}{{\mathbb{I}}}
\nc{\JJ}{{\mathbb{J}}}
\nc{\KK}{{\mathbb{K}}}
\nc{\LL}{{\mathbb{L}}}
\nc{\MM}{{\mathbb{M}}}
\nc{\NN}{{\mathbb{N}}}
\nc{\OO}{{\mathbb{O}}}
\nc{\PP}{{\mathbb{P}}}
\nc{\QQ}{{\mathbb{Q}}}
\nc{\RR}{{\mathbb{R}}}

\nc{\TT}{{\mathbb{T}}}
\nc{\UU}{{\mathbb{U}}}
\nc{\VV}{{\mathbb{V}}}
\nc{\WW}{{\mathbb{W}}}
\nc{\XX}{{\mathbb{X}}}
\nc{\YY}{{\mathbb{Y}}}
\nc{\ZZ}{{\mathbb{Z}}}

\DeclareMathOperator{\Id}{\mathrm{Id}}

\DeclareMathOperator{\ord}{\mathrm{ord}}

\nc{\Fun}{\mathrm{Fun}}

\DeclareMathOperator{\wind}{\mathrm{wind}}

\def\Vol{\mathrm{Vol}}

\DeclareMathOperator\Pic{\mathrm{Pic}}
\DeclareMathOperator{\Jac}{\mathrm{Jac}}
\DeclareMathOperator{\Div}{\mathrm{Div}}
\def\Bp{\mathsf{B}}

\def\m{\mathfrak{m}}

\DeclareMathOperator{\supp}{\mathrm{supp}}

\let\Re\relax
\let\Im\relax
\DeclareMathOperator{\Re}{\mathrm{Re}}
\DeclareMathOperator{\Im}{\mathrm{Im}}

\DeclareMathOperator{\Res}{\mathrm{Res}}
\DeclareMathOperator{\dbar}{\bar{\partial}}
\nc\chr[2]{\begin{bmatrix}#1 \\ #2\end{bmatrix}}

\def\vphi{\varphi}

\def\cst{\mathrm{cst}}

\def\hm{\mathrm{hm}}
\nc{\indic}{1\!\!1}

\DeclareMathOperator{\E}{\mathsf{E}}

\def\smm{\smallsetminus} 

\def\op{\mathrm{op}}

\DeclareMathOperator\id{\mathrm{id}}

\def\odd{\mathrm{odd}}

\def\diag{\mathrm{diag}}

\def\wtd{\widetilde}

\def\genbox#1#2#3#4#5#6{%
    \leavevmode\raise#4bp\hbox to#5bp{\vrule height#5bp depth0bp width0bp
    \pdfliteral{q .5 w \csname #2COLOR\endcsname\space RG
                       \csname #3PDF\endcsname{#5}{#6} S Q
             \ifx1#1 q \csname #2COLOR\endcsname\space rg 
                       \csname #3PDF\endcsname{#5}{#6} f Q\fi}\hss}}

\begin{document}

\title[Kenyon's identities for the height function and compactified free field]
{Kenyon's identities for the height function and\\ compactified free field in the dimer model}

\author[Mikhail Basok]{Mikhail Basok$^\mathrm{a}$}
\thanks{\textsc{${}^\mathrm{A}$ University of Helsinki, Department of Mathematics and Statistics, Helsinki, Finland}}

\keywords{dimer model, compactified free field}

\maketitle

\begin{abstract} 
  In his seminal paper~\cite{KenyonConfInvOfDominoTilings} published in 2000 Kenyon developed a method to study the height function of the planar dimer model via discrete complex analysis tools. The core of this method is a set of identities representing height correlations through the inverse Kasteleyn operator. In a general setup, such as considered in~\cite{CLR1,CLR2}, scaling limits of these identities produce a set of correlation functions written in terms of a Dirac Green's kernel with unknown boundary conditions. It was proven in~\cite{CLR1} that, in a simply connected domain, these correlation functions always coincide with correlation functions of the Gaussian free field given that they satisfy some natural a priori assumptions. This was generalized to doubly connected domains in the recent work~\cite{ChelkakDeiman}, where correlations are shown to be the correlations of a sum of Gaussian free field and a discrete Gaussian component. We generalize this result further to arbitrary bordered Riemann surfaces.
\end{abstract}

\tableofcontents

\section{Introduction}

\label{sec:Introduction}

In his pioneering works~\cite{KenyonConfInvOfDominoTilings, KenyonGFF} Kenyon considered the dimer model sampled in a sequence of Temperleyan polygons $\Omega_k\subset \frac{1}{k}\ZZ^2$ that approximate a planar domain $\Omega\subset\CC$ and proved the following two results. Assuming simply connectedness of $\Omega$, %
the centered dimer height functions converge to the (properly scaled) Gaussian free field with Dirichlet boundary conditions~\cite{KenyonGFF}. For multiply connected $\Omega$ %
the height gaps between boundary components of $\Omega_k$ converge to unknown conformally invariant random variables~\cite{KenyonConfInvOfDominoTilings}. 

These works of Kenyon are among the first where discrete complex analysis has been used as a tool to prove conformal invariance of a 2d lattice model. The entry point of applying this tool to the dimer model is a set of determinant identities used by Kenyon to express dimer height correlations via the inverse Kasteleyn operator. For a planar bipartite graph with a dimer model on it (a discrete domain) and a fixed collection of disjoint edges $w_1b_1,\dots, w_mb_m$ such an identity can be written as
\begin{equation}
  \label{eq:KenIdent_discrete}
  \EE \prod_{j = 1}^m (\overline h(v_j) - \overline h(u_j)) = \det[1_{i\neq j} K^{-1}(b_i, w_j)]\prod_{j = 1}^m K(w_j,b_j)
\end{equation}
where $\overline h = h - \EE h$ is the centered dimer height function, $K$ is the Kasteleyn operator and $u_jv_j$ is the dual edge of $w_jb_j$. The link with discrete complex analysis originates from interpreting $K$ as a discrete \emph{Dirac operator}. Pushing this forward one can hope to identify the scaling limit of $K^{-1}$ with the Dirac Green's kernel which, in its turn, suggests the following expression for the scaling limit of~\eqref{eq:KenIdent_discrete}:
\begin{equation}
  \label{eq:KenIdent_continuous}
  \EE \prod_{j = 1}^md\overline h(z_j) = \sum_{s\in \{ \pm \}^m} \det[1_{i\neq j}f^{[s_i,s_j]}(z_i,z_j)]\prod_{j = 1}^m dz_j^{[s_j]}
\end{equation}
where $\overline h$ is the (conjectural) random field that appears in the scaling limit, $z_j^{[+]} = z_j,\ z_j^{[-]} = \bar z_j$ and $f^{[\pm,\pm]}$ are holomorphic and anti-holomorphic components of the Dirac Green's kernel.

Once one is able to make this line of arguments rigorous, one still has to identify the right-hand side of~\eqref{eq:KenIdent_continuous}. This can be done, for example, by identifying the boundary conditions of the Dirac operator inherited by the combinatorics of discrete domains along their boundaries. 
For instance, Temperleyan combinatorics considered in~\cite{KenyonConfInvOfDominoTilings, KenyonGFF} imposes Dirichlet boundary conditions. These boundary conditions make $f^{[\pm,\pm]}$ conformally covariant, which implies the conformal invariance of the correlations~\eqref{eq:KenIdent_continuous} and allows one to identify $\overline h$ with the Gaussian free field in the case of a simply connected $\Omega$. As shown in a recent work~\cite{nicoletti2025temperleyan}, if $\Omega$ is multiply connected, these boundary conditions also allow one to express $f^{[\pm,\pm]}$ explicitly via the classical Szeg\"o kernel on the Schottky double of $\Omega$~\cite[Chapter~VI]{Fay}. This allows to fully determine the field $\overline h$ and, in particular, to identify the distribution of boundary height gaps considered by Kenyon in~\cite{KenyonConfInvOfDominoTilings}.

Other examples of explicit boundary conditions were studied in~\cite{russkikh2018dimers, russkikh2020dominos}. In general, however, the boundary conditions of the dimer models in question may be quite nasty and do \emph{not} produce conformally covariant scaling limits $f^{[\pm,\pm]}$. Moreover, even in the case of subdomains of regular lattices, if the boundary combinatorics of discrete domains is chosen arbitrary, there is no guarantee for the kernels $K^{-1}$ to converge or even to be bounded. This reflects the well-known prediction make by Kenyon and Okounkov~\cite{KenyonOkounkov} and asserting that the limit of the centered dimer height functions $\overline h$ becomes conformally invariant only after a proper choice of the conformal structure determined by the geometry of the so-called ``limit shape'', i.e. the scaling limit of $\EE h$~\cite{astala2026dimer, kenyon2022gradient}. If the change of the conformal structure is non-trivial, one can expect the entries of $K^{-1}$ to blow up as the formulae~\eqref{eq:KenIdent_continuous} with the Green's kernel being holomorphic in the initial complex structure would lead to a contradiction. Such a behaviour of $K^{-1}$ is exhibited, for example, in the case of the dimer model sampled on the Aztec diamond~\cite{berggren2024perfect}.

Having said all that, the aforementioned approach leading to formulae~\eqref{eq:KenIdent_continuous} can still in principle be adapted even when the change of the conformal structure is non-trivial by employing the theory of t-embeddings~\cite{CLR1,CLR2, KLRR}. The idea is to determine a \emph{gauge} of the Kasteleyn operator $K$ that will make the kernel $K^{-1}$ bounded and that will carry the geometric information needed to reconstruct the conformal structure. As suggested by~\cite{CLR2}, a natural candidate for such a gauge is a \emph{Coulomb gauge} determining a \emph{perfect t-embedding} of the dimer graph. The results of~\cite{CLR2} show that, after such a gauge was applied, the resulting kernel $K^{-1}$ can be bounded in terms of the regularity of the embedding only. Regularity theory developed in~\cite{CLR1} and results of~\cite{CLR2} imply that subsequential scaling limits of the kernels $K^{-1}$ exist and can be identified with Dirac Green's kernels in the conformal structure of the limiting t-surface (the limit of the lifts of perfect t-embeddings to the Lorentz space $\RR^{2,2}$) with unknown boundary conditions, given that the limiting t-surface is \emph{Lorentz minimal}. As discussed in ~\cite[Section~4.2]{CLR2}, the latter is expected to be Lorentz minimal as soon as the dimer height function is expected to have Gaussian fluctuations in the limit. Basing on these constructions and replacing the initial $K^{-1}$ in~\eqref{eq:KenIdent_discrete} with the one obtained after applying the gauge, we can derive formulas~\eqref{eq:KenIdent_continuous} for the limits of height correlations with $f^{[\pm,\pm]}$ being components of a Dirac Green's kernel on a Lorentz minimal surface in the Lorentz space $\RR^{2,2}$.

Perfect t-embeddings have been shown to exist and exhibit the expected properties in a number of particular cases (see~\cite{ChelkakRamassamy, berggren2024perfect, berggren2025perfect, berggren2024lozenge, keating2025perfect} and~\cite[Section~3.3]{KLRR}), yet the problem of their existence for general planar bipartite graphs with positive weights remains widely open. Nevertheless, the discussion above significantly strengthens the motivation to study the structure that formulas~\eqref{eq:KenIdent_continuous} impose on the correlations of the limiting field $\overline h$. While the functions $f^{[\pm,\pm]}$ that appear in these formulas may not be conformally covariant and sometimes depend on the particular subsequence along which the kernels $K^{-1}$ converge, the limiting field $\overline h$ is always expected to belong to a universal family of conformally invariant random generalized functions corresponding to compactified free fields in $\Omega$. This poses the problem of determining which properties of $f^{[\pm,\pm]}$ are sufficient to guarantee the conformal invariance of the right-hand side of~\eqref{eq:KenIdent_continuous} and, moreover, to identify the field $\overline h$ with a canonical object such as compactified free field.

In the case of simply connected domains a compelling answer to this question is given by~\cite[Lemma~7.3]{CLR1} that asserts the following. Denote the right-hand side of~\eqref{eq:KenIdent_continuous} by $U_m$. By definition $U_m$ is a harmonic differential in each of the variables. Assume that $U_2$ and $U_3$ have Dirichlet boundary conditions (which one expects as $\overline h$ is supposed to vanish outside of the domain) and correct singularities in the bulk: that is, $U_2$ has a quadratic pole with a prescribed residue along the diagonal and $U_3$ is regular. Then there exists a \emph{gauge} relating $f^{[\pm,\pm]}$ to the components of the Dirac Green's kernel satisfying ``spinor-type'' boundary conditions, and, in particular, $U_m$'s coincide with multi-point correlations of the gradient of the Gaussian free field.

When $\Omega$ is multiply connected with boundary components $B_0,B_1,\dots, B_n,\ n\geq 1,$ %
the well-known heuristics~\cite[Lecture~24]{gorin2021lectures} suggest that the limit field $\overline h$ should be of the form
\begin{equation}
  \label{eq:cff_domain}
  \overline h = \phi + \sum_{j = 1}^n (c_j - \EE c_j)\hm_j
\end{equation}
where $\phi$ is the Gaussian free field with zero boundary conditions scaled such that $\EE\phi(x)\phi(y)\sim -\frac{1}{2\pi^2}\log|x-y|$ as $|x-y|\to 0$, deterministic functions $\hm_j$ are harmonic measures of components $B_j$, and $c = (c_1,\dots, c_n)\in \RR^n$ is a random vector independent of $\phi$ satisfying the following properties:
\begin{enumerate}
  \item There exists a \emph{shift} $c_0\in \RR^n$ such that $c - c_0$ is almost surely an integer vector.
  \item For a given $u\in \ZZ^n + c_0$ we have
    \begin{equation}
      \label{eq:Dirichlet_energy}
      \PP[c = u]\propto \exp\left( -\frac{\pi}{2}\sum_{i,j = 1}^n u_iu_j\int_\Omega \nabla \hm_i\cdot \nabla \hm_j \right).
    \end{equation}
\end{enumerate}
Thus, generalizing~\cite[Lemma~7.3]{CLR1} to multiply connected domains requires to determine the distribution of the discrete component $c$ additionally to the GFF component $\phi$. In the case of a doubly connected domain such generalization was very recently obtained in~\cite{ChelkakDeiman} using %
properties of the Weierstrass $\wp$-function. This, in particular, allowed the authors of~\cite{ChelkakDeiman} 
to study the dimer model on a cylinder with Temperleyan conditions of different `colors' on two boundaries, in which case the Dirac Green's kernel $f^{[\pm,\pm]}$ is \emph{not} conformally covariant (see Section~\ref{subsec:Cylinder with black and white boundaries}).

The goal of this work is %
to generalize this approach further: from simply and doubly connected setups to arbitrary multiply connected domains further to bordered Riemann surfaces. In particular, we show that the algebraic structure of Kenyon's identities~\eqref{eq:KenIdent_continuous} %
allows one to identify the form of the joint distribution of dimer height gaps between boundary components and thus provide a strong supporting evidence for~\cite[Conjecture~24.1]{gorin2021lectures}.

\subsection{Main result: multiply connected domains}
\label{subsec:Main results: multiply connected domains}

We first formulate our theorem in the particular case of multiply connected domains. We need several definitions. %
Let $U_m$ denote the right-hand side of~\eqref{eq:KenIdent_continuous} as before. Let us say that $h = \phi + \sum_{j = 1}^n c_j \hm_j$ is a compactified free field with the shift $c_0\in \RR^n$ if $\phi$ and $c$ are %
as described above. Recall %
that the right-hand side of~\eqref{eq:cff_domain} is $h - \EE h$. Given an integer vector $\m\in \ZZ^n$ define
\begin{equation}
  \label{eq:E_m}
  \E_\m[X] \coloneqq \Bigl(\EE[\exp(\pi i\m\cdot (c-c_0))] \Bigr)^{-1}\cdot \EE[X\exp(\pi i\m\cdot (c-c_0))].
\end{equation}
We say that a harmonic differential $u$ defined in a neighborhood of $\partial\Omega$ satisfies Dirichlet boundary conditions if every primitive of $u$ extends continuously to $\partial\Omega$ and is locally constant there. We say that the differential $U_m$ satisfies Dirichlet boundary conditions it satisfies Dirichlet boundary conditions in every variable with other variables being fixed but arbitrary.

For various applications it is also convenient to allow functions $f^{[\pm,\pm]}$ to have multiplicative monodromies. (For instance, this  includes the setup of~\cite{ChelkakDeiman} in which case $f^{[\pm,\pm]}$ have monodromies $-1$ in each of the variables.) Let us say that $f^{[\pm,\pm]}(z_1,z_2)$ are multivalued with multiplicative monodromy if $f^{[\pm,\pm]}(z_1,z_2)$ are functions defined for $(z_1,z_2)$ from the universal cover of $\Omega\times\Omega$, and there exist characters $\chi_{1,2}:\pi_1(\Omega)\to \TT$ such that for every deck transformation $(\gamma_1,\gamma_2)\in \pi_1(\Omega)\times \pi_1(\Omega)$ 
\begin{equation}
  \label{eq:multivalued}
  (\gamma_1,\gamma_2)^\ast f^{[\pm,\pm]} = \chi_1(\gamma_1)\chi_2(\gamma_2)f^{[\pm,\pm]}.
\end{equation}
By abusing the notation we will still consider $f^{[\pm,\pm]}$ as functions on $\Omega\times\Omega$.

\renewcommand{\thethma}{A}
\begin{thma}
  \label{thma:domains}
  Assume that the functions $f^{[\pm,\pm]}$ defining $U_m$ are multivalued functions on $\Omega\times\Omega\smm\diag$ satisfying the following properties: 
  \begin{enumerate}
    \item All $U_m$ are single-valued.
    \item We have $\overline{f^{[+,+]}} = f^{[-,-]}$ and $\overline{f^{[+,-]}} = f^{[-,+]}$.
    \item $f^{[+,+]}$ is holomorphic in both variables while $f^{[+,-]}(z_1,z_2)$ is holomorphic in $z_1$ and anti-holomorphic in $z_2$.
    \item $U_2$ and $U_3$ satisfy Dirichlet boundary conditions, $U_3$ is extends continuously to $\Omega^3$ and $U_2(z_1,z_2) = \Re\frac{dz_1dz_2}{2\pi^2(z_1 - z_2)^2} + O(1)$ when $z_1$ is close to $z_2$ and $z_1,z_2$ stay at %
        a definite distance from $\partial\Omega$.
  \end{enumerate}
  Then there exist $c_0\in \RR^n$ and $\m\in \ZZ^n$ such that for the compactified free field with the shift $c_0$ we have $\EE[\exp(\pi i\m\cdot (c-c_0))]\neq 0$ and for each $z_1,\dots, z_m\in \Omega$
  \[
    U_m(z_1,\dots, z_m) = \E_\m \prod_{j = 1}^m d(h - \E_\m h)(z_j).
  \]
\end{thma}

Note that in order to pin $\m$ from Theorem~\ref{thma:domains} down to zero it is enough to assume, in addition, that $U_m$ define moments of a positive distribution which is usually known a priory. Moreover, it is shown in~\cite{ChelkakDeiman} that in the case of double connected domains it is enough to know that
\begin{equation}
  \label{eq:positivity}
  \E_\m (c_1^2) > (\E_\m c_1)^2.
\end{equation}
to conclude that $\m = 0$. A proper analogue of the criterion~\eqref{eq:positivity} would be the positive definiteness of the covariance matrix
\begin{equation}
  \label{eq:positivity_covariance}
  \E_\m (c_ic_j) - (\E_\m c_i)(\E_\m c_j),\quad i,j = 1,\dots, n.
\end{equation}
We hope to investigate this further in the future.

\subsection{Main result: bordered Riemann surfaces}
\label{subsec:Main results: bordered Riemann surfaces}

The next natural step after Theorem~\ref{thma:domains} is to consider the dimer model on a bordered Riemann surface. The obstacle for applying discrete complex analysis tools in this %
setup is that they require %
the dimer graphs to be adapted to the Euclidean metric on the plane. The way around this is to consider Riemann surfaces equipped with a locally flat metric with conical singularities. Dimer model in such a setup has been studied in the recent work~\cite{basok2023dimers} of the author as a part of a bigger research program %
initiated in~\cite{BerestyckiLaslierRayI,BerestyckiLaslierRayII}. The combination of these works provides a comprehensive description of the asymptotics of the dimer model sampled on Temperleyan graphs on Riemann surfaces (including multiply connected domains).

In general, if one applies Kenyon's method in the aforementioned setting, one ends %
up with a collection of functions $f^{[\pm,\pm]}$ which are multivalued and, furthermore, may have singularities of the form $z^\theta$ at conical singularities of the metric. The role of $dz$ in the correlation formulae~\eqref{eq:KenIdent_continuous} is %
played by a multivalued holomorphic differential $\omega$ defined so %
that $|\omega|^2$ is the locally flat singular metric on the surface. Kenyon's method produces harmonic differentials
\begin{equation}
  \label{eq:Um_surface}
  U_m(z_1,\dots, z_m) = \sum_{s\in \{ \pm \}^m} \det[1_{i\neq j}f^{[s_i,s_j]}(z_i,z_j)]\prod_{j = 1}^m d\omega^{[s_j]}(z_j)
\end{equation}
where $\omega^{[+]} = \omega$ and $\omega^{[-]} = \bar\omega$. These differentials correspond to height correlations evaluated with respect to a \emph{sign-indefinite} measure. The reason for the signs in the measure to occur lies in the way how the Kasteleyn theorem (and hence the combinatorial identities~\eqref{eq:KenIdent_discrete}) changes when the topology is non-trivial, see~\cite{Cimasoni} and references therein.

To define a proper notion of a compactified free field let us consider an open Riemann surface $\Sigma_0$ with boundary components $B_0,\dots, B_n,\ n\geq 0,$ and $g_0$ handles. We say that a harmonic differential $u$ on $\Sigma_0$ has Dirichlet boundary conditions if each boundary component $B_j$ has a neighborhood in which %
the primitive $\int u$ defines a single-valued harmonic function that is constant on $B_j$.
\begin{defin}
  \label{defin:integer_differential}
  Let $u$ be a harmonic differential on $\Sigma_0$ with Dirichlet boundary conditions. We say that $u$ is integer if for every path $\gamma$ connecting two boundary components we have $\int_\gamma u\in \ZZ$. Note that in this case for every loop $\gamma$ we must have $\int_\gamma u \in \ZZ$ as well.
\end{defin}
Generalizing the definition given above for %
planar domains we define a compactified free field $h$ on $\Sigma_0$ to be a multivalued random distribution that can be written as
\[
  h = \phi + \int\psi
\]
where $\phi$ is the Gaussian free field with zero boundary conditions normalized as before and $\psi$ is a random real-valued harmonic differential on $\Sigma_0$ satisfying the following properties:
\begin{enumerate}
  \item $\psi$ has Dirichlet boundary conditions and there exists a harmonic differential $\psi_0$ such that $\psi -\psi_0$ is almost surely integer as defined in Definition~\ref{defin:integer_differential}.
  \item If $u$ is a harmonic differential satisfying the previous assumption, then
    \begin{equation}
      \label{eq:Dirichlet_energy_surface}
      \PP[\psi = u]\propto \exp\left( -\frac{\pi}{2}\int_{\Sigma_0}u\wedge \ast u \right)
    \end{equation}
    where $\ast$ is the Hodge star (see Section~\ref{subsec:differential_forms}).
\end{enumerate}

Let us now define the sign-indefinite measure mentioned above. Assume that the genus $g_0$ of $\Sigma_0$ is positive. Let $A_{n+1}, \ldots, A_{n+g_0}, B_{n+1},\ldots, B_{n+g_0}$ be %
a set of loops on $\Sigma_0$ that complements $B_1,\dots, B_n$ to a simplicial basis in $H_1(\Sigma_0,\ZZ)$ (see Section~\ref{subsec:Schottky double} for more details). Given a harmonic differential $u$ on $\Sigma_0$ satisfying Dirichlet boundary conditions define
\begin{equation}
  \label{eq:def_of_Q0_differentials}
  Q_0(u) = \sum_{j = n+1}^{n+g_0} \int_{A_j}u\int_{B_j}u
\end{equation}
Finally, denote by $l_1,\dots, l_n$ simple paths connecting the boundary component $B_0$ to the boundary components $B_1,\dots, B_n$ and not intersecting $A_{n+1}, \ldots, A_{n+g_0}, B_{n+1},\ldots, B_{n+g_0}$. Given an integer vector $\m = (m_1,\dots, m_n, a_1,\dots, a_{g_0}, b_1,\dots, b_{g_0})\in \ZZ^{n + 2g_0}$ define
\begin{equation}
  \label{eq:def_of_Qm}
  Q_\m(u) = Q_0(u) - \sum_{j = 1}^n m_j\int_{l_j}u + \sum_{j = 1}^{g_0} \left(a_j \int_{B_{n+j}}u - b_j\int_{A_{n+j}}u\right).
\end{equation}
Similarly to~\eqref{eq:E_m}, given the compactified free field $h = \phi + \int \psi$ with a shift $\psi_0$ we define
\begin{equation}
  \label{eq:E_m_surfaces}
  \E_\m[X] \coloneqq \Bigl(\EE[\exp(\pi iQ_\m(\psi - \psi_0))] \Bigr)^{-1}\cdot \EE[X\exp(\pi iQ_\m(\psi-\psi_0))].
\end{equation}
We can finally formulate our main result.

\renewcommand{\thethma}{B}
\begin{thma}
  \label{thma:surfaces}
  Fix an arbitrary collection of points $x_1,\dots, x_N\in \Sigma_0$ and define $\Sigma_0' = \Sigma_0\smm\{ x_1,\dots, x_N \}$. Assume that the functions $f^{[\pm,\pm]}$ and the holomorphic differential $\omega$ %
  used to define $U_m$ in~\eqref{eq:Um_surface} are multivalued on $\Sigma_0'\times\Sigma_0'\smm\diag$ and $\Sigma_0'$ respectively and satisfy the following properties: 
  \begin{enumerate}
    \item The differentials $U_m$ are single-valued.
    \item We have $\overline{f^{[+,+]}} = f^{[-,-]}$ and $\overline{f^{[+,-]}} = f^{[-,+]}$.
    \item $f^{[+,+]}$ is holomorphic in both variables while $f^{[+,-]}(z_1,z_2)$ is holomorphic in $z_1$ and anti-holomorphic in $z_2$.
    \item $U_2$ and $U_3$ extend continuously to the whole $\Sigma_0\times\Sigma_0\smm\diag$ and $\Sigma_0\times\Sigma_0\times\Sigma_0$ respectively, satisfy Dirichlet boundary conditions and we have
      \[
        U_2(p_1,p_2) = \Re\frac{dz(p_1)dz(p_2)}{2\pi^2(z(p_1) - z(p_2))^2} + O(1)
      \]
      if $z$ is a local coordinate, $p_1$ is close to $p_2$, and $p_1,p_2$ stay at %
      a definite distance from $\partial\Omega$.
  \end{enumerate}
  Then there exist $\m\in \ZZ^{n+2g_0}$ and a harmonic differential $\psi_0$ satisfying Dirichlet boundary conditions such that for the compactified free field with the shift $\psi_0$ we have $\EE[\exp(\pi iQ_\m(\psi - \psi_0))]\neq 0$ and for each $p_1,\dots, p_m\in \Omega$
  \[
    U_m(p_1,\dots, p_m) = \E_\m \prod_{j = 1}^m d(h - \E_\m h)(p_j).
  \]
\end{thma}

\subsection{Discussion of the proof}
\label{subsec:Discussion of the proof}

The proof of Theorems~\ref{thma:domains} and~\ref{thma:surfaces} is based on three ingredients. The first is a formula analogous to~\eqref{eq:KenIdent_continuous} where on the left-hand side $\EE$ is replaced by $\E_\m$ with unknown $\m$, the field $h$ is a compactified free field with some shift, and on the right-hand side $f^{[\pm,\pm]}$ is replaced by entries of the \emph{Cauchy kernel} of a perturbed Cauchy--Riemann operator acting on sections of a spin line bundle on the \emph{Schottky double} of the Riemann surface. We derive this formula from a variation identity for the observable $F_\m(v) = \E_\m\exp\left(i\int v\wedge dh\right)$ derived in~\cite{basok2023dimers}, see Theorem~\ref{thmas:variation_of_F} and Corollary~\ref{cor:correlations}. In the case of a planar multiply connected domain such a formula can also be derived from~\cite[Theorem~4.2]{berggren2025gaussian}.

The second ingredient is a square root trick that appeared in~\cite[Theorem~3.8]{ChelkakDeiman}, see Lemma~\ref{lemma:extension}. This trick allows to derive a duality between (yet unknown) boundary conditions of $f^{[\pm,\pm]}(p,q)$ in the first and in the second variables respectively. This in its turn allows to interpret these functions as a Cauchy kernel of a certain (also unknown) line bundle on the Schottky double.

The last ingredient is an abstract argument asserting that a line bundle on a Riemann surface has a Cauchy kernel if and only if it is isomorphic to a tensor product of a spin line bundle and a flat line bundle. This implies that there is a gauge relating $f^{[\pm,\pm]}$ with a Cauchy kernel from the first ingredient, thus identifying the correlations. The point in the Jacobian parametrizing the flat line bundle determines the shift of the compactified free field and the choice of $\m$. We demonstrate how this works in practice via a number of examples in Section~\ref{sec:Examples}.

To conclude the discussion, let us note that Theorem~\ref{thma:surfaces} holds also for non-bordered surfaces. In fact, it is much easier in this case as one does not need to deal with boundary conditions of $f^{[\pm,\pm]}$, thus the second (the most technical) step can be omitted and the theorem becomes a simple corollary of the calculations made in Section~\ref{subsec:Fermionization of correlations of a compactified free field}.

\subsection{Organization of the paper}
\label{subsec:Organization of the paper}

The methods that we use require a number of basic facts from the theory of Riemann surface. In hope of making the paper as self-contained as possible we have included preliminary Section~\ref{sec:Preliminaries on Riemann surfaces} where most of these facts are set out. 
Seasoned readers can jump straight to Section~\ref{sec:Compactified free field}, which
contains some preliminary results about compactified free field obtained in~\cite{basok2023dimers} and their immediate corollaries. Section~\ref{sec:proof} contains the proof of the main result. Let us emphasize that the proof, although looking rather abstract on the first glance, allows to compute the parameters from Theorem~\ref{thma:domains} and~\ref{thma:surfaces} in many particular cases. This is demonstrated in Section~\ref{sec:Examples} dedicated to examples.

\subsection*{Acknowledgment}
The author thanks Dmitry Chelkak and Zach Deiman for bringing their very recent research~\cite{ChelkakDeiman} to author's attention before publication and for enlightening discussions of the doubly connected setup. The work was supported by Research Council of Finland grants 333932 and 363549.

\section{Preliminaries on Riemann surfaces}
\label{sec:Preliminaries on Riemann surfaces}

The purpose of this section is to remind some basic facts and constructions from the theory of Riemann surfaces that we will need. We will try to keep the presentation as self-contained as possible, though some of the proofs will have to be omitted. All the facts below are fairly standard and can be found in standard textbooks such as~\cite{MumfordTata1, GriffitsHarris, Forster, Fay}; we also address the reader to~\cite[Appendix]{basok2023dimers} and references therein.

Throughout this section $\Sigma$ denotes a compact Riemann surface without a boundary and $g$ denotes its genus.

\subsection{Harmonic differentials and Hodge decomposition}
\label{subsec:differential_forms}

Every differential form of degree 1 (a 1-form for brevity) $\omega$ on a Riemann surface can be locally written as 
\[
  \omega = f\,dz + g\,d\bar z
\]
where $z$ is a holomorphic local coordinate and $f,g$ are functions. Globally this defines a decomposition $\omega = \omega^{1,0} + \omega^{0,1}$ where $\omega^{1,0}$ is obtained from $\omega$ by forgetting the term $g\,d\bar z$ (one can easily show that such projection is consistent with a coordinate change). The forms $\omega^{1,0}$ and $\omega^{0,1}$ are called $(1,0)$- and $(0,1)$-forms respectively.

A 1-form $\omega$ is called holomorphic if it is a $(1,0)$-form and can be locally written as $f\,dz$ where $f$ is holomorphic. The $(0,1)$-form $\bar \omega$ is then called anti-holomorphic. A 1-form is called harmonic if it is a linear combination of holomorphic and anti-holomorphic forms. Note that if a harmonic 1-form $u$ is real-valued, then $u = u^{1,0} + \overline{u^{0,1}} = 2\Re u^{1,0}$, i.e. every real-valued harmonic 1-form is a real part of a holomorphic one.

The Hodge star $\ast$ is a linear operation on 1-forms defined such that $\ast\ast = -1$ and $\ast \omega = i\bar\omega$ if $\omega$ is a $(1,0)$-form. Note that if $f$ is a smooth function on $\Sigma$, then 
\[
  \int_\Sigma df\wedge \ast df
\]
is the Dirichlet energy of $f$. A 1-form is called exact if it can be written as $df$ for some smooth function $f$, and co-exact if it can be written as $\ast df$. The Hodge decomposition theorem asserts that every smooth 1-form $u$ can be uniquely decomposed as
\begin{equation}
  \label{eq:Hodge_decomposition}
  u = u_h + df + \ast dh
\end{equation}
where $u_h$ is harmonic. By taking $(0,1)$ components of both sides of~\eqref{eq:Hodge_decomposition} we conclude that every smooth $(0,1)$ form $\alpha$ can be uniquely decomposed as
\begin{equation}
  \label{eq:Hodge_decomposition_01}
  \alpha = \alpha_h + \dbar F
\end{equation}
where $\alpha_h$ is anti-holomorphic. This is called the Dolbeault decomposition of $\alpha$.

Note that every harmonic 1-form $u$ is closed, i.e. $du = 0$. In particular, one can define a pairing between the space of harmonic 1-forms and the first homology group $H_1(\Sigma, \ZZ)$: if $u$ is a harmonic 1-form and a loop $\gamma$ represents a homology class from $H_1(\Sigma,\ZZ)$, then
\begin{equation}
  \label{eq:pairing}
  \langle u, [\gamma]\rangle = \int_\gamma u.
\end{equation}
A standard fact from the Hodge theory says that this pairing is non-degenerate. In particular, the dimension of the space of harmonic 1-forms is $2g$. This implies that the (complex) dimension of the space of holomorphic 1-forms is $g$.

\subsection{Simplicial basis in $H_1(\Sigma,\ZZ)$ and normalized differentials}
\label{subsec:Simplicial basis and normalized differentials}

Recall that $H_1(\Sigma,\ZZ)$ has a non-degenerate skew-symmetric bilinear form given by the algebraic intersection number. We can choose a basis of cycles $A_1,\dots, A_g,B_1,\dots, B_g\in H_1(\Sigma, \ZZ)$ such that for each $i,j$
\[
  A_i\cdot A_j = B_i\cdot B_j = 0,\qquad A_i\cdot B_j = \delta_{ij}.
\]
Such basis is called a simplicial basis.

Given a simplicial basis as above we can fix holomorphic differentials $\omega_1,\dots, \omega_g$ such that
\[
  \int_{A_i}\omega_j = \delta_{ij},\qquad i,j = 1,\dots,g.
\]
Such differentials form a basis in the space of holomorphic differentials. They are called normalized holomorphic differentials. Define the $g\times g$ matrix $\Bp$ by
\[
  \Bp_{ij} = \int_{B_i}\omega_j.
\]
The matrix $\Bp$ is called the matrix of b-periods. It is symmetric and $\Im \Bp$ is positively defined.

Given a basis $A_1,\dots, A_g,B_1,\dots, B_g\in H_1(\Sigma,\ZZ)$ we can identify the space of real harmonic differentials on $\Sigma$ with $\RR^{2g}$ by
\[
  u\mapsto (\int_{A_1}u,\ldots,\int_{A_g}u,\int_{B_1}u,\ldots, \int_{B_g}u)\in \RR^{2g}.
\]
If the basis is simplicial, this identification intertwines the wedge pairing on differentials with the standard anti-symmetric form on $\RR^{2g}$:
\begin{equation}
  \label{eq:bilinear_relations}
  \int_\Sigma u\wedge v = \sum_{j = 1}^g\left(\int_{A_j}u\int_{B_j}v - \int_{B_j}u\int_{A_j}v\right).
\end{equation}

\subsection{Line bundles on Riemann surfaces}

A (holomorphic) line bundle $\Ll$ over $\Sigma$ is a complex manifold together with a projection $\pi: \Ll\to \Sigma$ which is a holomorphic map and satisfies the following assumptions:
\begin{itemize}
  \item For each $p\in \Sigma$ the fiber $\pi^{-1}(p)$ has a structure of one-dimensional vector space over $\CC$.
  \item The fibration $\pi$ is locally trivial, that is, for each $p\in \Sigma$ there exists a neighborhood $U\subset \Sigma$ of it and a biholomorphic map $F: \pi^{-1}(U)\to U\times \CC$ such that the diagram
\begin{equation}
  \label{eq:line_bundle}
  \vcenter{\xymatrix{\pi^{-1}(U) \ar[rd]^\pi \ar[rr]^F && U\times \CC\ar[dl] \\
                                                     & U}}
\end{equation}
is commutative and for each $p$ the restriction of $F$ to the corresponding fibers is an isomorphism of vector spaces. The map $F$ is called a trivialization of $\Ll$ over $U$.
\end{itemize}

A morphism between two line bundles $\Ll_1,\Ll_2$ is a holomorphic map $F:\Ll_1\to \Ll_2$ that commutes with the projection to $\Sigma$ and acts as a linear map on each fiber. A morphism is an isomorphism if it has an inverse. A line bundle is called trivial if it is isomorphic to the line bundle $\Sigma\times\CC$. The restriction $\Ll\vert_U$ of a line bundle $\Ll$ to $U\subset \Sigma$ is the line bundle $\pi:\pi^{-1}(U)\to U$ over $U$.

It is convenient to represent line bundles via transition functions. Fix an open cover $(U_i)_{i\in I}$ of $\Sigma$ and assume that a line bundle $\Ll$ is trivial over each $U_i$ (such cover exists by definition). Fix a trivialization $F_i$ of $\Ll\vert_{U_i}$. Then for each $i,j\in I$ there exists a holomorphic function $\vphi_{ij}:U_i\cap U_j\to \CC^\ast$ such that 
\[
  F_i(v) = \vphi_{ij}(\pi(v))\cdot F_j(v)
\]
for each $v\in \Ll\vert_{U_i\cap U_j}$, where on the right-hand side we multiply the second coordinate of $F_j$ only. The functions $\vphi_{ij}$ are called transition functions. It is easy to see that $(\vphi_{ij})_{i,j\in I}$ determine $\Ll$ up to an isomorphism. Note that $\vphi_{ij} = \vphi_{ji}^{-1}$ and $\vphi_{ij}\vphi_{jk}\vphi_{ki} = 1$. Given any collection $(\vphi_{ij})_{i,j\in I}$ satisfying these two properties we can form a line bundle with $(\vphi_{ij})_{i,j\in I}$ being its transition functions.

Let us give some important examples of line bundles. Complex structure on $\Sigma$ allows to define the holomorphic cotangent bundle traditionally denoted by $K_\Sigma$. We define it using transition functions. Fix an open cover $(U_i)_{i\in I}$ as above and let for each $i$ the map $z_i:U_i\to \CC$ be a local holomorphic coordinate. Put 
\[
  \vphi_{ij} = \frac{dz_j}{dz_i},
\]
that is, $\vphi_{ij}$ is the derivative of $z_j$ with respect to $z_i$. Chain rule ensures that $(\vphi_{ij})_{i,j\in I}$ satisfies the properties of transition function, thus they define a line bundle $K_\Sigma$.

Another example of a line bundle can be constructed as follows. Let $p\in \Sigma$ be a point and $z$ be a local coordinate defined in a neighborhood $U_1$ of $p$. Assume that $z(p) = 0$ and put $U_2 = \Sigma\smm\{ p \}$. Define the line bundle $\Oo_\Sigma(p)$ by declaring the transition function $\vphi_{12}$ to be
\[
  \vphi_{12}(x) = z(x).
\]

Given two line bundles $\Ll_1,\Ll_2$ its tensor product $\Ll_1\otimes\Ll_2$ is obtained by tensor multiplying $\Ll_1,\Ll_2$ fiberwise. Note that if $\Ll_{1,2}$ has transition functions $\vphi^{1,2}_{ij}$ defined with respect to the same cover, then $\vphi^1_{ij}\vphi^2_{ij}$ are transition functions of $\Ll_1\otimes \Ll_2$.

Let $\Ll$ be a lint bundle. The dual line bungle $\Ll^\ast$ is obtained from $\Ll$ by replacing each fiber by the dual vector space. If $\Ll$ has transition functions $\vphi_{ij}$, then $\vphi_{ij}^{-1}$ are transition functions of $\Ll^\ast$. Note that $\Ll\otimes \Ll^\ast$ is trivial. 
The Picard group of $\Sigma$ is the group $\Pic(\Sigma)$ whose elements are isomorphism classes of line bundles, the product is given by tensor multiplication and the inversion is given by $\Ll\mapsto \Ll^\ast$.

Finally, a morphism between line bundles $\pi_1:\Ll_1\to \Sigma,\ \pi_2:\Ll_2\to \Sigma$ is any map $\Phi:\Ll_1\to\Ll_2$ such that $\pi_2\circ\Phi = \pi_1$ and for each $p\in \Sigma$ the map $\Phi$ restricted to the fiber of $\Ll_1$ over $p$ acts as a linear map between $\Ll_1\vert_p$ and $\Ll_2\vert_p$. If $\Ll_1,\Ll_2$ is represented by transition functions $\vphi_{ij}^{1,2}$, then a morphism can be represented by a collection of maps $\Phi_i: U_i\to \CC$ satisfying
\[
  \Phi_i\vphi^{(1)}_{ij} = \vphi^{(2)}_{ij}\Phi_j.
\]
A morphism is smooth if the map $\Phi$ is smooth and holomorphic if $\Phi$ is holomorphic. A morphism is an isomorphism if it admits an inverse.

\subsection{Sections of a line bundle}
\label{subsec:Sections of a line bundle}

A section $f$ of a line bundle $\Ll$ over a set $U\subset \Sigma$ is a map $f: U\to \Ll$ such that $\pi\circ f = \id$, that is, for each $p\in U$ the $f(p)$ is a point in the fiber over $p$. Sections over $U$ can be added together and multiplied by a function on $U$. A section is called smooth (resp. holomorphic) if $f$ is a smooth (resp. holomorphic) map. A global section is a section over $\Sigma$. Note that since fibers of a line bundle are vector spaces, each fiber has a distinguished point corresponding to the origin. In particular, each line bundle has a global zero section that sends each point of $\Sigma$ to the origin. This section is denoted by $0$. Let us put
\[
  \begin{split}
    &\mC^\infty(U,\Ll) = \{ \text{smooth sections of $\Ll$ over }U \},\\
    &H^0(U,\Ll) = \{ \text{holomorphic sections of $\Ll$ over }U \}.
  \end{split}
\]
Note that $\mC^\infty(U,\Ll)$ and $H^0(U,\Ll)$ are vector spaces. Both spaces have the zero section $0$ as the origin.

If $f_1,f_2$ are sections of line bundles $\Ll_1$ and $\Ll_2$, then $f_1f_2$ is naturally a section of $\Ll_1\otimes \Ll_2$. If $f$ is a section of $\Ll$, then $f^{-1}$ is a section of $\Ll^\ast$ outside of the zero locus of $f$.

Note that sections of the bundle $\Sigma\times \CC$ are just functions on $\Sigma$. Moreover, let $\Ll$ be a line bundle with transition functions $\vphi_{ij}$ corresponding to a cover $(U_i)_{i\in I}$. Then any section $f$ of $\Ll$ over $U$ can be represented as a collection of functions $f_i: U_i\cap U\to \CC$ such that
\[
  f_i(p) = \vphi_{ij}(p) f_j(p),\qquad p\in U_i\cap U_j\cap U.
\]
A change of trivializations of $\Ll$ over $U_i$ amounts in a gauge transformation of $f_i$-s: each $f_i$ is multiplied by a holomorphic function $\vphi_i:U_i\to \CC^\ast$. A section is smooth/holomorphic if $f_i$ are smooth/holomorphic.

A section $f$ of $\Ll$ is a meromorphic section over $U$ if it is a section of $\Ll$ over $U$ minus a discrete set and the functions $f_i$ representing $f$ are meromorphic on $U$. For example, if $f$ is a holomorphic section of $\Ll$, then $f^{-1}$ is a meromorphic section of $\Ll^\ast$. It is easy to see that the notion of a meromorphic section does not depend on the choice of the cover and trivializations. Moreover, locations and orders of poles and zeros of meromorphic sections are also invariant. In particular, we can define the divisor of a global meromorphic section as a formal linear combination of points from $\Sigma$:
\[
  \Div(f) = \sum_{p\in \Sigma} \ord_pf \cdot p.
\]
The degree of this divisor is the sum of the coefficients:
\[
  \deg \Div(f) = \sum_{p\in \Sigma} \ord_pf.
\]

Let us consider some examples. It is easy to see that sections of the holomorphic cotangent bundle $K_\Sigma$ are complex-valued 1-forms on $\Sigma$ that can locally be written as $f(z)\,dz$ where $z$ is a local coordinate, that is, sections of $K_\Sigma$ are $(1,0)$-forms. Holomorphic sections of $K_\Sigma$ are then holomorphic 1-forms, also called holomorphic differentials.

A line bundle is trivial if and only if it has a non-zero holomorphic section. If $\Ll$ has such a section $f$, then the map $\Sigma\times \CC\ni (p,z)\mapsto zf(p)\in \Ll$ defines an isomorphism. The line bundle $\Oo_\Sigma(p)$ has a canonical global section $1_p$ that is represented by $f_1,f_2$ where
\[
  f_1 = z,\qquad f_2 = 1.
\]
The section $1_p$ has a simple zero at $p$. Consequently, the line bundle $\Oo_\Sigma(-p) = \Oo_\Sigma(p)^\ast$ has a meromorphic section $1_p^{-1}$ with a simple pole at $p$.

Let now $\Ll$ be a line bundle. Define
\begin{equation}
  \label{eq:L(p)}
  \Ll(p) = \Ll\otimes \Oo_\Sigma(p),\qquad \Ll(-p) = \Ll\otimes \Oo_\Sigma(-p).
\end{equation}
We have injective linear maps
\[
  H^0(\Sigma, \Ll(-p))\to H^0(\Sigma, \Ll)\to H^0(\Sigma, \Ll(p)),\qquad f\mapsto f\cdot 1_p.
\]
Using the first map one can identify $H^0(\Sigma, \Ll(-p))$ with the space of sections of $\Ll$ vanishing at $p$. Given a holomorphic section $f$ of $\Ll(p)$ one can form a meromorphic section of $\Ll$ by considering $f\cdot 1_p^{-1}$. This section has a simple pole at $p$ and no other poles. We conclude that $H^0(\Sigma, \Ll(p))$ can be identified as the space of meromorphic sections of $\Ll$ with a simple pole at $p$.

\subsection{Metric and curvature. Riemann--Roch theorem}
\label{subsec:Metric and curvature. Riemann--Roch theorem}

Let $\Ll$ be a line bundle. A Hermitian metric on $\Ll$ is a choice of a Hermitian metric $|\cdot|^2$ in each fiber such that for each smooth section $f$ of $\Ll$ the function $|f|^2$ is smooth. Using that $\Ll$ is locally trivial and a partition of unity it is easy to show that a Hermitian metric always exists.

Assume that $\vphi_{ij}$ are transition functions of $\Ll$ corresponding to a cover $(U_i)_{i\in I}$. Then any Hermitian metric on $\Ll$ can be represented by a collection of positive functions $\rho_i: U_i\to \RR$ such that 
\begin{equation}
  \label{eq:rho_i}
  \rho_i = |\vphi_{ij}|^{-2}\rho_j.
\end{equation}
Given a collection of such functions and a section $f$ represented by functions $f_i:U_i\to \CC$ we can calculate $|f(p)|^2$ as $\rho_i(p)|f_i(p)|^2$ if $p\in U_i$; the condition above ensures that this does not depend on $i$.

Recall that a (complex valued) 1-form on $\Sigma$ evaluated at a given point $p\in \Sigma$ is a linear map from the tangent space $T_p\Sigma$ to $\CC$. Similarly, a 1-form with values in a line bundle $\Ll$ evaluated at a given point $p\in \Sigma$ is a linear map from $T_p\Sigma$ to the fiber $\Ll\vert_p$. An example of such a 1-form is $f \omega$ where $f$ is a section of $\Ll$ and $\omega$ is a usual 1-form. Every 1-form with values in $\Ll$ can be written as a linear combination of the 1-forms such as above.

Let $|\cdot|^2$ be a Hermitian metric on $\Ll$ represented by functions $\rho_i$ as above. In each chart $U_i$ consider the differential operator
\[
  f_i\mapsto df_i + f\,\partial \log \rho_i.
\]
It is easy to see that these operators are invariant under a change of trivializations of $\Ll$ over $U_i$, thus they glue together to a well-defined differential operator $\nabla$ that sends each $\mC^\infty$ section $f$ to a 1-form $\nabla f$ with values in $\Ll$. Note that $\nabla$ is linear and satisfies the Leibniz rule: if $h$ is a smooth function, then $\nabla (hf) = fdh + h\nabla f$. Such differential operators are called connections. The connection $\nabla$ has two important additional properties:
\begin{itemize}
  \item Compatibility with complex structure: $\nabla f = 0$ if and only if $f$ is holomorphic.
  \item Compatibility with the metric: $d|f|^2 = 2\Re\langle \nabla f,f\rangle$, where $\langle\nabla f, f\rangle$ evaluated at a given point $p$ is a function on the tangent space $T_p\Sigma$ that sends a tangent vector to $\Ll\vert_p$ via $\nabla f$ and then scalar multiplies the image by $f(p)$.
\end{itemize}
Such connection is called the metric connection.

Besides the metric connection, a Hermitian metric $|\cdot|^2$ on $\Ll$ has the metric curvature form. On $U_i$ consider the differential form
\[
  \Theta_i = \dbar\! \partial\log \rho_i.
\]
It is easy to check that $\Theta_i$ do not depend on the trivializations, thus, they together form a well-defined 2-form $\Theta$ called the curvature form of the metric $|\cdot|^2$.

Assume for the moment that $\Ll$ has a holomorphic section $f$ that does not vanish (i.e. $\Ll$ is trivial). Then $\dbar\!\partial \log |f|^2 = \Theta$, thus $\int_\Sigma\Theta = 0$ by the Stokes theorem. If $f$ is a meromorphic section, then the $\dbar\!\partial \log |f|^2 = \Theta$ outside of the divisor of $f$ and it is easy to see that 
\[
  \int_\Sigma \Theta = -2\pi i \deg \Div(f).
\]
In fact, it can be shown (e.g. by proving that every line bundle admits a meromorphic section) that for every line bundle $\Ll$ the number
\[
  \deg\Ll = \frac{i}{2\pi} \int_\Sigma \Theta
\]
is integer and does not depend on the metric. This number is called the degree of $\Ll$.

For example, the $\deg \Oo_\Sigma(p) = 1$ and $\deg K_\Sigma = 2g-2$. The last relation is nothing but the Gauss--Bonnet theorem: indeed, it is easy to see that in this case
\[
  \Theta = \frac{i}{2}K\cdot d\Vol,
\]
where $K$ is the Gauss curvature and $d\Vol$ is the volume form determined by Riemannian metric corresponding to the metric on $K_\Sigma$.

The Riemann--Roch theorem is an important formula that allows to estimate the dimension of the space of holomorphic sections of a line bundle. It asserts that
\begin{equation}
  \label{eq:Riemann-Roch}
  \dim H^0(\Sigma, \Ll) = \deg \Ll - g + 1 + \dim H^0(\Sigma, K_\Sigma\otimes \Ll^\ast).
\end{equation}

\subsection{Flat line bundles and Cauchy--Riemann operators}
\label{subsec:Flat line bundles and Cauchy--Riemann operators}

Recall that $\Pic(\Sigma)$ is the set of isomorphism classes of line bundles on $\Sigma$ endowed with the group structure corresponding to the tensor product. It is easy to see that the function $\Ll\mapsto \deg\Ll$ defines a homomorphism from $\Pic(\Sigma)$ to $\ZZ$. Let us focus on the kernel of this homomorphism, i.e. on the set $\Pic^0(\Sigma)$ of line bundles of degree 0. One way to classify them is by looking at families of Cauchy--Riemann operators.

Given an anti-holomorphic 1-form $\alpha$ we can now construct a line bundle $\Ll_\alpha$ as follows. Fix a cover $(U_i)_{i\in I}$ of $\Sigma$ by simply-connected sets. For each $j\in I$ fix a base point $p_j\in U_j$ and define functions
\[
  \Phi_j(p) = \exp(-2i\int_{p_j}^p\alpha_j),\qquad \vphi_{ij} = \Phi_j/\Phi_i,
\]
where the integration is taken inside $U_j$. By definition for each $i,j$ the function $\vphi_{ij}$ is constant and $|\vphi_{ij}| = 1$. Denote by $\Ll_\alpha$ the holomorphic line bundle determined by $\vphi_{ij}$. The fact that $\vphi_{ij}\in \TT$ implies that $\Ll_\alpha$ admits a metric which is flat, i.e. whose curvature vanishes. In particular, $\deg \Ll_\alpha = 0$.

Note that functions $\Phi_j$ put together define a smooth isomorphism $\Phi$ between $\Ll_\alpha$ and the trivial line bundle $\CC\times \Sigma$. We emphasize that this morphism is smooth but not holomorphis, thus $\Ll_\alpha$ is not necessary trivial. By definition, $\Phi$ sends sections of $\Ll_\alpha$ to smooth functions. Given a smooth section of $\Ll_\alpha$ represented by a collection of functions $f_i:U_i \to \CC$ the corresponding smooth function $f$ is given point-wise by $f(p) = \Phi_j(p)f_j(p)$. Note that
\[
  (\dbar + \alpha)f = \vphi_j \dbar f_j
\]
over each $U_j$. We conclude that $\Phi$ intertwines $\dbar$ acting on $\mC^\infty(\Sigma, \Ll_\alpha)$ with $\dbar + \alpha$ acting on $\mC^\infty(\Sigma)$.

\begin{lemma}
  \label{lemma:flat_line_bundles}
  Assume that $\Ll$ is a holomorphic line bundle on $\Sigma$ and $\deg \Ll = 0$. Then there exists an anti-holomorphic $\alpha$ such that $\Ll\cong \Ll_\alpha$.
\end{lemma}

Recall that $\deg$ is a homomorphism from $\Pic(\Sigma)$ to $\ZZ$. We conclude that if $\deg\Ff = \deg\Ff_0$, then there exists a unique $[\alpha]\in \Jac(\Sigma)$ such that $\Ff\cong \Ff_0\otimes \Ll_\alpha$. In particular, for any $d\in \ZZ$ isomorphism classes of bundles of degree $d$ are all isomorphic to bundles of the form $\Ff_0\otimes \Ll_\alpha$ for an arbitrary fixed $\Ff_0$ with $\deg \Ff_0 = d$. For example, we can pick $\Ff_0 = \Oo_\Sigma(dp) \coloneqq \Oo_\Sigma(p)^{\otimes d}$ for some $p\in \Sigma$.

\subsection{Jacobian of a Riemann surface}
\label{subsec:Jacobian of a Riemann surface}

We have observed that degree zero line bundles have the form $\Ll_\alpha$; the next natural question is which $\alpha$-s define isomorphic bundles. Note that $\Ll_\alpha\otimes\Ll_\beta^\ast\cong \Ll_{\alpha - \beta}$, thus out questions boils down to a description of $\alpha$-s such that $\Ll_\alpha$ is trivial. By definition the latter is true if and only if $\Ll_\alpha$ admits a non-vanishing global holomorphic section. Such section would correspond to a non-zero function $f\in \mC^\infty(\Sigma)$ such that $(\dbar + \alpha) f = 0$. Such function exists if and only if for each loop $\gamma$ on $\Sigma$ we have
\[
  \int_\gamma\Im \alpha \in \pi\ZZ
\]
in which case we can put $f(p) = \exp(-2i\int_{p_0}^p\Im \alpha)$. Denote the space of anti-holomorphic differentials by $\Omega^{0,1}_\Sigma$ and the subspace of differentials satisfying the relation above by $\Lambda\subset \Omega_\Sigma^{0,1}$. We conclude that degree zero line bundles are classified by the factor space
\begin{equation}
  \label{eq:def_of_Jac}
  \Jac(\Sigma)\coloneqq \frac{\Omega_\Sigma^{0,1}}{\Lambda}
\end{equation}
Since $\Ll_\alpha\otimes \Ll_\beta\cong\Ll_{\alpha+\beta}$ we conclude that the established correspondence between $\Pic^0(\Sigma)$ and $\Jac(\Sigma)$ is a group isomorphism.

Recall that $\Omega_\Sigma^{0,1}$ is a $g$-dimensional vector space, thus $\Jac(\Sigma)$ possess a structure of a complex manifold. Another point of view on $\Jac(\Sigma)$ is provided by the one-to-one correspondence between anti-holomorphic differentials $\alpha$ and the harmonic differential 
\begin{equation}
  \label{eq:def_of_psi}
  \psi = 2\pi^{-1}\Im\alpha.
\end{equation}
Choosing a simplicial basis $A_1,\dots, A_g,B_1,\dots, B_g$ in $H_1(\Sigma,\ZZ)$ we can further identify $\psi$ with the vector
\[
  (\int_{A_1}\psi,\dots, \int_{A_g}\psi,\int_{B_1}\psi,\dots, \int_{B_g}\psi)\in \RR^g
\]
which induces the smooth diffeomorphism
\begin{equation}
  \label{eq:Jac_torus}
  \Jac(\Sigma) \cong \frac{\RR^{2g}}{(2\ZZ)^{2g}}.
\end{equation}
Finally, notice that $\psi$ induces a character $\chi:H_1(\Sigma,\ZZ)\to \TT$ given by
\begin{equation}
  \label{eq:def_of_chi}
  \chi(\gamma)= \exp(\pi i \int_\gamma\psi)
\end{equation}
and that $\chi$ is trivial if and only if $\psi$ corresponds to the identity in $\Jac(\Sigma)$. It is worth noting that all the above mentioned correspondences are isomorphisms of Abelian groups.

\begin{defin}
  \label{defin:Jac_coord}
  In what follows we will often be switching between the different perspective on $\Jac(\Sigma)$ described above, that is, we will often be saying that $\alpha$, $\psi$, $(a,b)\in \RR^{2g}$ or $\chi$ represent the same point in $\Jac(\Sigma)$ meaning that $\psi = 2\pi^{-1}\Im\alpha$, $(a,b)$ are the $A$- and $B$-periods of $\psi$, $\chi$ is defined by~\eqref{eq:def_of_chi} etc.
\end{defin}

\begin{rem}
  \label{rem:multivalued_functions}
  Consider a smooth section of $\Ll_\alpha$ identified with a smooth function $f$. Fix a base point $p_0\in \Sigma$ and consider the multivalued function $f_\alpha(p) = \exp(2i\int_{p_0}^p\Im\alpha)f(p)$. Note that the monodromy of $f$ is given by the character $\chi$ representing the same point in the Jacobian as $\alpha$. Note also that $f$ represents a holomorphic section if and only if $\dbar f_\alpha = 0$. Thus, we the map $f\mapsto f_\alpha$ provides a way to identify sections of $\Ll_\alpha$ with multivalued functions in such a way that the Cauchy--Riemann operators acting on the former and on the latter are intertwined.
\end{rem}

\subsection{Cauchy kernel of a line bundle}
\label{subsec:Cauchy kernel of a line bundle}

Let $\Ll$ be a line bundle and $f\in \mC^\infty(\Sigma, \Ll)$. Fix an open cover $(U_i)_{i\in I}$ and let $\vphi_{ij}$ be transition functions corresponding to $\Ll$. Let $f_i:U_i\to \CC$ represent $f$. Since $\vphi_{ij}$ are holomorphic, $(0,1)$-forms $\dbar f_i$ form together a $(0,1)$-form with values in $\Ll$. We denote this 1-form by $\dbar f$. Taking completions with respect to suitable Sobolev norms allows to treat $\dbar$ as a Fredholm operator. If this operator is invertible, then the inverse has a kernel that is called a Cauchy kernel of $\Ll$.

To define such kernel formally we need the following auxiliary notation. Given a holomorphic map $f:X\to Y$ and a line bundle $\pi:\Ll\to Y$ we can define $f^\ast\Ll$ by
\[
  f^\ast\Ll = \{ (v,x)\in \Ll\times X\ \mid\ \pi(v) = f(x) \}.
\]
It is straightforward to verify that $f^\ast\Ll$ has a structure of a line bundle on $X$. Let $p,q: \Sigma\times\Sigma\to \Sigma$ be the projections onto the first and the second factors respectively. Given two line bundles $\Ll_1,\Ll_2$ on $\Sigma$ define
\begin{equation}
  \label{eq:boxtimes}
  \Ll_1\boxtimes\Ll_2\coloneqq p^\ast\Ll_1\otimes q^\ast\Ll_2.
\end{equation}
Sections of $\Ll_1\boxtimes\Ll_2$ can be viewed as objects of the form $C(p,q),\ p,q\in \Sigma,$ which behave as sections of $\Ll_1$ in $p$ and as sections of $\Ll_2$ in $q$.

Let us say that $\Dd^{-1}(p,q)$ is the Cauchy kernel for $\dbar$ acting on smooth sections of $\Ll$ if $\Dd^{-1}$ is a meromorphic section of $\Ll\boxtimes (K_\Sigma\otimes\Ll^\ast)$ with a simple pole along the diagonal and such that for each $f\in \mC^\infty(\Sigma, \Ll)$
\begin{equation}
  \label{eq:def_of_Dd-1}
  f(p) = \int_\Sigma \Dd^{-1}(p,q)\wedge \dbar f(q).
\end{equation}
To make sense of the integrand in the expression above, note that for a fixed $p$ the section $\Dd^{-1}(p,q)$ behaves as a $(1,0)$-form with values in $\Ll^\ast$, and $\dbar f$ is a $(0,1)$-form with values in $\Ll$. Thus we can use the pairing between $\Ll^\ast$ and $\Ll$ to interpret the expression $\Dd^{-1}(p,q)\wedge\dbar f(q)$ as a 2-form on $\Sigma$ which makes the integral in~\eqref{eq:def_of_Dd-1} well-defined.

Let $U\subset \Sigma$ be an open set such that there exist a $\zeta\in H^0(U,\Ll)$ that never vanishes. Then $\zeta^{-1}\in H^0(U,\Ll^\ast)$ naturally. Let $z$ be a local coordinate on $U$
\[
  \Dd^{-1}(p,q) = C(p,q)\,\zeta(p)\zeta(q)^{-1}\,dz(q),\qquad p,q\in U
\]
where $C(p,q)$ is a meromorphic function on $U\times U$ with a simple pole on a diagonal. We can expand
\[
  C(p,q) = \frac{a(q)}{z(q) - z(p)} + O(1),\qquad p\to q.
\]
Define
\[
  \Res_{p=q} \Dd^{-1}(p,q) \coloneqq a(q).
\]
It is easy to see that $\Res_{p=q} \Dd^{-1}(p,q)$ does not depend on the choice of $\zeta$ and $z$, that the relation~\eqref{eq:def_of_Dd-1} is equivalent to
\[
  \Res_{p=q}\Dd^{-1}(p,q) = \frac{1}{2\pi i}.
\]
In particular, note that if $\Dd^{-1}$ is the Cauchy kernel for $\Ll$, then $\wtd\Dd^{-1}(p,q) = \Dd^{-1}(q,p)$ is the Cauchy kernel for $K_\Sigma\otimes\Ll^\ast$.

\begin{lemma}
  \label{lemma:Dd-1_exists}
  Let $\Ll$ be a holomorphic line bundle on $\Sigma$. If the Cauchy kernel of $\Ll$ exists, then $\deg\Ll = g-1$ and $\dim H^0(\Sigma,\Ll) = 0$. Moreover, if the Cauchy kernel exists, then it is unique.
\end{lemma}
\begin{proof}
  Let $\Dd^{-1}$ denote the Cauchy kernel. The fact that $\Dd^{-1}$ is the left inverse to $\dbar$ implies that $\dbar$ has zero kernel, that is, $\dim H^0(\Sigma, \Ll) = 0$. Riemann--Roch formula~\eqref{eq:Riemann-Roch} implies that
  \[
    \dim H^0(\Sigma, K_\Sigma\otimes\Ll^\ast) = \dim H^0(\Sigma, K_\Sigma\otimes\Ll^\ast) - \dim H^0(\Sigma, \Ll) = g-1-\deg\Ll
  \]
  whereas
  \[
    \deg\Ll\leq g-1.
  \]
  As we observed above, $\wtd\Dd^{-1}(p,q) = \Dd^{-1}(q,p)$ is the Cauchy kernel for $K_\Sigma\otimes \Ll^\ast$. Acting as previously we conclude that
  \[
    \deg (K_\Sigma \otimes \Ll^\ast) = 2g-2 - \deg\Ll \leq g-1.
  \]
  Combining these two inequalities we conclude that $\deg \Ll = g-1$.

  Finally, if $\Dd_1^{-1}$ is another Cauchy kernel, then, for a generic $q$, the difference $\Dd^{-1}(p,q) - \Dd_1^{-1}(p,q)$ is a non-zero holomorphic section of $\Ll$ which is a contradiction.
\end{proof}

\subsection{Spin line bundles}
\label{subsec:Spin line bundles}

By definition, a spin line bundle $\Ff$ is a line bundle on $\Sigma$ for which there exists an isomorphism $\Ff^{\otimes 2}\cong K_\Sigma$. Note that in this case
\[
  \deg \Ff = \tfrac{1}{2} \deg K_\Sigma = g-1.
\]
Spin line bundles exist: indeed, one can pick an arbitrary $\Ff_0$ of degree $g-1$ and $\alpha\in \Omega_\Sigma^{0,1}$ such that $\Ff_0^{\otimes 2}\cong K_\Sigma\otimes \Ll_\alpha$ (see Section~\ref{subsec:Flat line bundles and Cauchy--Riemann operators} for the definition of $\Ll_\alpha$). Then $\Ff_0\otimes \Ll_{-\alpha/2}$ is a spin line bundle. Moreover, if $\Ff_1,\Ff_2$ are two line bundles, then $\Ff_1\otimes \Ff_2^\ast\cong \Ll_\alpha$ where $\alpha$ represents a root of unity in the Jacobian, that is, $2\alpha\in \Lambda$ (see Section~\ref{subsec:Jacobian of a Riemann surface}). This implies that there are $2^{2g}$ isomorphism classes of spin line bundles.

Spin line bundles can be also understood topologically. Let $\zeta$ be a smooth section of $\Ff$ with isolated zeros and $\gamma$ be a smooth oriented loop on $\Sigma$ such that $\zeta\vert_\gamma$ does not vanish. Identifying $\zeta^2$ with a 1-form and applying it to the tangent vector field on $\gamma$ we obtain a map $f:\gamma\to \CC^\ast$. Let $\wind(\gamma, \zeta^2)\in \ZZ$ be the index of this map. It is easy to see that $\wind(\gamma, \zeta^2)\mod 2$ does not depend on the choice of $\zeta$ and thus defines an invariant of $\Ff$. Note that if $\gamma$ is contractible we have $\wind(\gamma, \zeta^2)\mod 2 = 1$. Put
\begin{equation}
  \label{eq:def_of_q}
  q_\Ff(\gamma) = \wind(\gamma, \zeta^2) + 1\mod 2.
\end{equation}
One can show that $q_\Ff(\gamma)$ depends only on the homology class of $\gamma$, that is, $q_\Ff$ is a map from $H_1(\Sigma, \ZZ)$ to $\ZZ/2\ZZ$. A fundamental theorem of Johnson asserts that such a map $q$ is related to a spin line bundle if and only if it is quadratic, that is, for each $\gamma_1,\gamma_2\in H_1(\Sigma, \ZZ)$
\begin{equation}
  \label{eq:quadratic_relation}
  q(\gamma_1+\gamma_2) = q(\gamma_1) + q(\gamma_2) + \gamma_1\cdot\gamma_2 \mod 2.
\end{equation}
Let $A_1,\dots, A_g,B_1,\dots,B_g$ be a simplicial basis in $H_1(\Sigma, \ZZ)$. Clearly, the values of $q_\Ff$ on this basis determine $q_\Ff$ uniquely. In particular, we have exactly $2^{2g}$ possible maps $q$ which matches with the number of non-isomorphic spin line bundles.

\subsection{Family of Cauchy--Riemann operators}
\label{subsec:Family of Cauchy--Riemann operators}

Recall the notation $\Ll_\alpha$ from Section~\ref{subsec:Flat line bundles and Cauchy--Riemann operators}. Let $A_1,\dots, A_g,B_1,\dots,B_g$ be a simplicial basis in $H_1(\Sigma, \ZZ)$ and let $q_0:H_1(\Sigma, \ZZ)\to \ZZ/2\ZZ$ be the quadratic function such that
\[
  q(A_i) = q(B_i) = 0,\qquad i = 1,\dots, g
\]
(note that $q\neq 0$ in general), and let $\Ff_0$ be the corresponding spin line bundle. We call $\Ff_0$ the spin line bundle with zero characteristics. Note that the definition of $\Ff_0$ depends on the choice of the simplicial basis. 
Pick an antiholomorphic $\alpha$ and recall that we have a natural smooth isomorphism $\Phi:\Ll_\alpha \to \CC\times \Sigma$ intertwining $\dbar$ acting on $\mC^\infty(\Sigma,\Ll_\alpha)$ with the operator $\dbar+\alpha$ acting on $\mC^\infty(\Sigma)$. Consequently, the smooth isomorphism $\Id\otimes\Phi:\Ff_0\otimes\Ll_\alpha\to \Ff_0$ intertwines $\dbar$ acting on $\mC^\infty(\Sigma,\Ff_0\otimes\Ll_\alpha)$ with $\dbar+\alpha$ acting on $\mC^\infty(\Sigma,\Ff_0)$. Denote the latter operator by $\Dd_\alpha$. The Cauchy kernel $\Dd^{-1}$ of $\Ff_0\otimes\Ll_\alpha$ (if exists) gets to be intertwined with the kernel of the inverse operator $\Dd_\alpha^{-1}$. 

Recall (see Section~\ref{subsec:Cauchy kernel of a line bundle}) that $\Dd^{-1}$ is defined as a meromorphic section of $(\Ff_0\otimes \Ll_\alpha)\boxtimes (K_\Sigma\otimes\Ff_0^\ast\otimes\Ll_\alpha^\ast)$ with a simple pole along the diagonal. Since $\Ff_0\otimes\Ff_0\cong K_\Sigma$ we have
\[
  (\Ff_0\otimes \Ll_\alpha)\boxtimes (K_\Sigma\otimes\Ff_0^\ast\otimes\Ll_\alpha^\ast)\cong (\Ff_0\otimes \Ll_\alpha)\boxtimes (\Ff_0\otimes \Ll_{-\alpha}).
\]
Consequently, $\Dd_\alpha^{-1}$ is a section of $\Ff_0\boxtimes \Ff_0$ with a singularity along the diagonal. The following remark will play an important role in the main proof:

\begin{rem}
  \label{rem:gauge}
  It is straightforward to verify that, given distinct $p_1,\dots, p_m\in \Sigma$ we can identify the expressions
  \begin{equation}
    \label{eq:Dd_vs_Ddalpha}
    \Dd^{-1}(p_1,p_2)\Dd^{-1}(p_2,p_3)\cdot\ldots \cdot \Dd^{-1}(p_m,p_1) = \Dd_\alpha^{-1}(p_1,p_2)\Dd_\alpha^{-1}(p_2,p_3)\cdot\ldots \cdot \Dd_\alpha^{-1}(p_m,p_1)
  \end{equation}
  by treating both of them as holomorphic differentials. To emphasize this fact we call the morphism $\Id\otimes\Phi$ a `gauge' relating $\Dd^{-1}$ with $\Dd_\alpha^{-1}$.
\end{rem}

\section{Compactified free field and family of Cauchy--Riemann operators}
\label{sec:Compactified free field}

\subsection{Schottky double of $\Sigma_0$}
\label{subsec:Schottky double}

Leg $\Sigma_0$ be a Riemann surface with $n+1$ boundary components $B_0,\dots, B_n$ and $g_0$ handles. Denote by $\Sigma_0^\op$ a copy of $\Sigma_0$ with reversed orientation. Put
\[
  \Sigma = \Sigma_0\cup\Sigma_0^\op
\]
with the boundaries of $\Sigma_0$ and $\Sigma_0^\op$ being identified in the natural way. There is a natural conformal structure on $\Sigma$ which is compatible with those on $\Sigma_0$ and $\Sigma_0^\op$. This makes $\Sigma$ a Riemann surface called Schottky double of $\Sigma_0$. The involution $\sigma:\Sigma\to \Sigma$ interchanging $\Sigma_0$ and $\Sigma_0^\op$ is anti-holomorphic. Note that $\partial\Sigma_0$ is precisely the set fixed points of $\sigma$. Denote by $g$ the genus of the Schottky double $\Sigma$. Note that
\[
  g = 2g_0 + n.
\]
We also fix a simplicial basis $A_1,\dots, A_g, B_1,\dots, B_g$ in $H_1(\Sigma, \ZZ)$ (see Section~\ref{subsec:Simplicial basis and normalized differentials}) as follows.
\begin{itemize}
  \item Take $B_1,\dots, B_n$ to be the homology classes represented by boundary components of $\Sigma_0$.
  \item Take $l_1,\dots, l_n$ to be simple paths connecting $B_0$ with $B_1,\dots, B_n$ respectively, and put $A_j = l_j\cup \overleftarrow{\sigma(l_j)},\ j = 1,\dots, n,$ where $\overleftarrow{l}$ is $l$ with reversed orientation.
  \item Complete this by $A_{n+1},\dots, A_{n+g_0}, B_{n+1},\dots, B_{n+g_0}\subset \Sigma_0$ in an arbitrary way and put $A_{n+g_0 + j} = -\sigma_\ast A_{n+j},\ B_{n+g_0+j} = \sigma_\ast B_{n+j},\ j = 1,\dots, g_0$.
\end{itemize}

The notion of Schottky double provides a convenient way to encode boundary conditions of differentials.

\begin{defin}
  \label{defin:bc}
  Let us say that a smooth 1-form $u$ on $\Sigma_0$ satisfies Dirichlet boundary conditions if it admits an smooth extension $u$ to $\Sigma$ such that $\sigma^\ast u = -u$. A differential satisfies Neumann boundary conditions if it admits a smooth extension $u$ to $\Sigma$ satisfying $\sigma^\ast u = u$.
\end{defin}

\subsection{Observables of a compactified free field}
\label{subsec:Observables of a compactified free field}

Let $v$ be a smooth 1-form on $\Sigma_0$ with Neumann boundary conditions and $h = \phi + \int\psi$ be a compactified free field defined as in Section~\ref{subsec:Main results: bordered Riemann surfaces}. Consider the random variable
\begin{equation}
  \label{eq:observable}
  \int_{\Sigma_0}v\wedge dh.
\end{equation}
In order to make a precise sense out of it we can smoothen the Gaussian free field $\phi$ by convoluting it with a smooth modifier; then the integral~\eqref{eq:observable} becomes well-defined and we may take the limit to define the actual random variable. Let
\[
  v = u + df_1 + \ast df_2
\]
be the Hodge decomposition of $v$ on $\Sigma$ (see Section~\ref{subsec:differential_forms}). By the uniqueness of the Hodge decomposition we have
\[
  \sigma^\ast u = u,\qquad \sigma^\ast f_1 = f_2,\qquad \sigma^\ast f_2 = -f_2.
\]
It is easy to see (this can also be taken as a definition)
\begin{equation}
  \label{eq:observable_decomposed}
  \int_{\Sigma_0}v\wedge dh = \int_{\Sigma_0}\phi\,d\ast df_2 + \int_{\Sigma_0}u\wedge \psi.
\end{equation}

Recall that for each $\m = (m_1,\dots, m_n,a_1,\dots, a_{g_0}, b_1,\dots, b_{g_0})\in \ZZ^{n+2g_0}$ we defined the twisted expectation $\E_\m$ and the function $Q_\m$, see Section~\ref{subsec:Main results: bordered Riemann surfaces}. By abusing the notation we denote by $\m$ the harmonic differential on $\Sigma$ satisfying $\sigma^\ast\m = \m$ and such that 
\begin{equation}
  \label{eq:periods_of_m}
  \int_{B_j}\m = m_j,\quad j = 1,\dots, n,\qquad \int_{A_{n+j}}\m = a_j,\ \int_{B_{n+j}}\m = b_j,\quad j = 1,\dots,g_0.
\end{equation}
Note that other periods of $\m$ are determined by the symmetry $\sigma^\ast \m = \m$. Bilinear relations~\eqref{eq:bilinear_relations} imply that
\begin{equation}
  \label{eq:Q_symmetric}
  Q_\m(u) = Q_0(u) + \int_{\Sigma_0} \m\wedge u.
\end{equation}

\begin{defin}
  \label{defin:observable}
  Assume that $\m$ is such that $\E_\m$ is well-defined. Define the function $F_\m$ on the space of smooth 1-form $v$ on $\Sigma_0$ with Neumann boundary conditions by
  \[
    F_\m(v) = \E_\m \exp\left(i\int_{\Sigma_0}v\wedge dh\right).
  \]
\end{defin}

\begin{rem}
  \label{rem:dependence_on_psi0}
  We emphasize that when $\Sigma_0$ is a multiply connected domain we have $\E_0 = \EE$ and $F_0(v)$ is just the characteristic function of $h$. On the other hand, when $g_0>0$ the expectation $\E_\m$ is always an expectation with respect to a sign indefinite measure since the quadratic term $Q_0$ is non-trivial. Moreover, in this case $\E_\m$ depends on the choice of $\psi_0$ only modulo $2\ZZ$: adding an integer differential $\m_1$ to $\psi_0$ is equivalent to replacing $\E_\m$ with $\E_{\m + \tilde\m_1}$ where $\tilde\m_1$ is defined such
  \[
    \begin{split}
      &\int_{A_j}\tilde\m_1 = \int_{B_j}\tilde\m_1 = 0,\qquad j = 1,\dots, n,\\
      &\int_{A_j}\m_1 = \int_{A_j}\tilde\m_1,\ \int_{B_j}\m_1 = \int_{B_j}\tilde\m_1,\ j = n+1,\dots, n+g_0,\\
      &\sigma*\tilde\m_1 = \tilde\m_1.
    \end{split}
  \]
\end{rem}

\subsection{Variation of the observable}
\label{subsec:Variation of the observable}

Following the notation from Section~\ref{subsec:Family of Cauchy--Riemann operators} denote by $\Ff_0$ the spin line bundle on the double $\Sigma$ corresponding to the quadratic function $q_0: H_1(\Sigma,\ZZ)\to \ZZ/2\ZZ$ such that
\begin{equation}
  \label{eq:def_of_q0}
  q_0(A_i) = q_0(B_i) = 0\qquad \forall i
\end{equation}
where $A_1,\dots,A_g,B_1,\dots, B_g$ is the simplicial basis fixed in Section~\ref{subsec:Schottky double}. Given a smooth $(0,1)$-form $\beta$ on $\Sigma$ define
\[
  \Dd_\beta = \dbar + \beta
\]
to be acting on $\mC^\infty(\Sigma,\Ff_0)$ (the space of smooth sections of $\Ff_0$). Applying the Dolbeault decomposition to $\beta$ (see Section~\ref{subsec:differential_forms}) we can write
\[
  \beta = \alpha + \dbar \vphi
\]
where $\alpha$ is anti-holomorphic and $\vphi$ is a complex-valued smooth function. Note that
\begin{equation}
  \label{eq:Dd_beta=eDd_alphae}
  \Dd_\beta = e^{-\vphi} \Dd_\alpha e^\vphi,
\end{equation}
thus the inverting kernel of $\Dd_\beta$ is given by
\begin{equation}
  \label{eq:Dd_beta_inv}
  \Dd_\beta^{-1}(p,q) = e^{\vphi(q) - \vphi(p)}\Dd_\alpha^{-1}(p,q).
\end{equation}
In particular, $\Dd_\beta^{-1}$ exists if and only if the Cauchy kernel $\Dd_\alpha^{-1}$ discussed in Section~\ref{subsec:Family of Cauchy--Riemann operators} exists.

Assume now additionally that $\sigma^\ast\beta = -\bar\beta$. We have $\sigma^\ast \Im\beta = \Im \beta$, that is, $\Im \beta$ restricted to $\Sigma_0$ has Neumann boundary conditions and we can consider $F_\m(2\Im\beta)$. It was shown in~\cite{basok2023dimers} that the logarithmic variation of $F_\m(2\Im \beta)$ has a close relation with the near diagonal asymptotics of the Cauchy kernel of a certain $\dbar$ operator. Let us recall this relation.

Fix a simply connected coordinate chart $U$ on $\Sigma$ with a holomorphic local coordinate $z$. By~\eqref{eq:Dd_beta_inv} and the properties of $\Dd_\alpha^{-1}$ discussed in Section~\ref{subsec:Family of Cauchy--Riemann operators} we can write when $p,q\in U$
\begin{multline}
  \label{eq:Dd_beta_inv_expansion}
  \Dd_\beta(p,q) = \\
  = e^{\vphi(q) - \vphi(p) + 2i\int_p^q\Im\alpha} \left[ \frac{1}{2\pi i(z(q) - z(p))} + R_\alpha(z(q)) + O(z(q) - z(p)) \right]\,dz(p)^{1/2}dz(q)^{1/2}
\end{multline}
where $dz^{1/2}$ is a holomorphic section of $\Ff_0$ over $U$ whose square is equal to $dz$. A straightforward check shows that 
\begin{equation}
  \label{eq:def_of_r}
  r_\alpha(q)\coloneqq R_\alpha(z(q))\,dz(q)
\end{equation}
does not depend on the choice of the coordinate $z$, thus it is a globally defined holomorphic differential on $\Sigma$.

Consider now a compactified free field $h = \phi + \int\psi$. Fix a harmonic differential $\psi_0$ such that $\psi-\psi_0$ is integer. Define
\begin{equation}
  \label{eq:def_of_alpha0}
  \alpha_0 = \pi i\psi_0^{0,1}
\end{equation}
where $\psi_0^{0,1}$ is the $(0,1)$ part of $\psi_0$ (see Section~\ref{subsec:differential_forms}). Note that $\alpha_0$ is an anti-holomorphic 1-form and we have
\begin{equation}
  \label{eq:psi0_vs_alpha0}
  \psi_0 = 2\pi^{-1}\Im\alpha_0.
\end{equation}
Since $\psi_0$ has Dirichlet boundary conditions the form $\alpha_0$ extends to $\Sigma$ in such a way that $\sigma^\ast\alpha_0 = \bar\alpha_0$.

\begin{lemma}
  \label{lemma:Em_exists}
  Assume that $\m$ is a harmonic differential with integer periods and such that $\sigma^\ast\m = \m$. Let $\alpha_\m = \pi i \m^{0,1}$. Then $\EE \exp(\pi i Q_\m(\psi - \psi_0))\neq 0$ if and only if the Cauchy kernel of the line bundle $\Ff_0\otimes \Ll_{\alpha_\m + \alpha_0}$ exists.
\end{lemma}
\begin{proof}
  This follows from the expression for $\EE \exp(\pi i Q_\m(\psi - \psi_0))$ obtained in~\cite[Lemma~7.2]{basok2023dimers} and relating it with an appropriate theta constant, and from Riemann theorem on theta divisor, see~\cite[Theorem~9.3]{basok2023dimers}.
\end{proof}

\begin{thmas}
  \label{thmas:variation_of_F}
  Let $\beta = \beta(t)$ be a differentiable family of smooth $(0,1)$-forms on $\Sigma$ such that $\sigma^\ast\beta(t) = -\overline{\beta(t)}$ for each $t$ and $\Dd_{\beta(t)}^{-1}$ exists for each $t$. Let $\beta(t) = \alpha(t) + \dbar\vphi(t)$ be the Dolbeault decomposition. Assume that $\m$ is a harmonic differential with integer periods and such that $\sigma^\ast\m = \m$. Let $\alpha_\m = \pi i\m^{0,1}$, note that $\sigma^\ast\alpha_\m = -\bar\alpha_\m$. We have
  \begin{equation}
    \label{eq:variation_identity}
    \frac{d}{dt}\log F_\m(2\Im\beta(t)) = \int_\Sigma\left( r_{\alpha(t) +\alpha_\m + \alpha_0} + \tfrac{1}{\pi i}\partial\Re\vphi(t) \right)\wedge \dot{\beta}(t)
  \end{equation}
  where $\dot{\beta}$ denotes the derivative with respect to $t$.
\end{thmas}
\begin{rem}
  \label{rem:integration_over_Sigma}
  Note that we have $\sigma^\ast\alpha(t) = -\overline{\alpha(t)}$ and $\sigma^\ast\alpha_0 = \bar\alpha_0$. It can be shown that this implies $\sigma^\ast r_{\alpha(t) +\alpha_\m + \alpha_0} = \bar r_{-\alpha(t) - \alpha_\m + \alpha_0}$. This implies that the phase of the integral on the right-hand side of~\eqref{eq:variation_identity} is non-trivial, thus $F(2\Im\beta(t))$ is not real-valued.
\end{rem}
\begin{proof}
  A straightforward calculation shows that the differential $-2\pi i r_\alpha$ coincides with the differential $W_\alpha$ introduced in the proof of~\cite[Lemma~4.4]{basok2023dimers}. The variation identity now follows from the variation identity given in the proof of~\cite[Lemma~4.4]{basok2023dimers} and the identity given in~\cite[Corollary~7.1]{basok2023dimers}.
\end{proof}

\begin{rem}
  \label{rem:detDelta}
  Fix an arbitrary Hermitian metric on $\Ff_0$ and $\Sigma$ and consider the Laplace operator $\Delta_\beta = \Dd_\beta^\ast\Dd_\beta$. Comparing~\eqref{eq:variation_identity} and the classical Quillen's variation identity~\cite{Quillen} one can show that 
  \[
    \det_\zeta\Delta_{\beta+\alpha_0} = C_0|F_0(2\Im \beta)|^2
  \]
  where $C_0>0$ is some constant independent of $\beta$. Note also that $F_\m(v) = \cst F_0(v + \alpha_\m)$.
\end{rem}

\subsection{Fermionization of correlations of a compactified free field}
\label{subsec:Fermionization of correlations of a compactified free field}

Theorem~\ref{thmas:variation_of_F} allows to give a short proof of determinant expressions for correlations of a compactified free field $h$. In the case of multiply connected domains such expressions appeared in~\cite{nicoletti2025temperleyan} (as one can easily see by using the standard formula for $\Dd^{-1}_\alpha$ in terms of theta functions~\cite[Chapter~VI]{Fay}).

\begin{lemma}
  \label{lemma:variation_of_r}
  Let $\beta(t)$ be a differentiable family of smooth $(0,1)$-forms on $\Sigma$ such that $\Dd_{\beta(t)}^{-1}$ exists for each $t$. Let $\beta(t) = \alpha(t) + \dbar\vphi(t)$ be the Dolbeault decomposition. Assume that $p\in \Sigma$ is such that $p\notin \supp \beta(t)$ for each $t$. Then we have
  \begin{equation}
    \label{eq:ddt_r}
    \frac{d}{dt}(r_{\alpha(t)}(p) + \tfrac{1}{\pi i}\partial \Re\vphi(p)) = -\int_\Sigma\Dd_{\beta(t)}^{-1}(p,q)\Dd_{\beta(t)}^{-1}(q,p)\wedge\dot{\beta}(q,t) + \frac{1}{2\pi i}\overline{\dot{\beta}(p,t)}.
  \end{equation}
  where $\dot{\beta}$ is the derivative of $\beta$ with respect to $t$.
\end{lemma}
\begin{proof}
  General calculus of operators implies that for each $p_1\neq p$ we have
  \[
    \frac{d}{dt}\Dd^{-1}_{\beta(t)}(p_1,p) = -\int_\Sigma \Dd_{\beta(t)}^{-1}(p_1,q)\Dd_{\beta(t)}^{-1}(q,p)\wedge\dot{\beta}(q,t).
  \]
  The lemma now follows from~\eqref{eq:Dd_beta_inv_expansion} and~\eqref{eq:def_of_r}.
\end{proof}

\begin{thmas}
  \label{thmas:determinantal_identity}
  Let $h = \phi + \int\psi$ be a compactified free field and $\psi_0 = 2\pi^{-1}\Im\alpha_0$ be such that $\psi - \psi_0$ is integer. Assume that $\m$ is a harmonic differential with integer periods and such that $\sigma^\ast\m = \m$. Let $\alpha_\m = \pi i\m^{0,1}$ and let $v_1,\dots, v_m$ be smooth 1-forms on $\Sigma_0$ whose supports do not intersect each other and the boundary of $\Sigma_0$. Then we have
  \begin{equation}
    \label{eq:determinantal_identity}
    \E_\m \prod_{j = 1}^m \int_{\Sigma_0}v_j\wedge dh = \int_{\Sigma^m}\det[1_{j\neq k} \Dd^{-1}_{\alpha_0 + \alpha_\m}(p_j,p_k) + 1_{j=k} r_{\alpha_0 + \alpha_\m}(p_j)]\wedge \prod_{j = 1}^m v_j(p_j)
  \end{equation}
  where for $\det[1_{j\neq k} \Dd^{-1}_{\alpha_0 + \alpha_\m}(p_j,p_k) + 1_{j=k} r_{\alpha_0 + \alpha_\m}(p_j)]$ is interpreted as a $(1,0)$-form in each variable by identifying $\Ff_0\otimes\Ff_0$ with the holomorphic cotangent bundle $K_\Sigma$, and $\wedge$ means the wedge product in each variable.
\end{thmas}
\begin{proof}
  Note that
  \begin{equation}
    \label{eq:crl1}
    \E_\m  \prod_{j = 1}^m \int_{\Sigma_0}v_j\wedge dh = (-i)^m\frac{\partial^m}{\partial t_1\partial t_2\ldots \partial t_m} F_\m(t_1v_1 + \ldots + t_mv_m)\vert_{t=0}.
  \end{equation}
  Let $\beta_j = iv_j^{0,1}$, so that we have
  \[
    v_j = 2\Im \beta_j.
  \]
  Note that $\supp\beta_j\subset \supp v_j$. Applying Dolbeault decomposition we can write
  \[
    \beta_j = \alpha_j + \dbar \vphi_j.
  \]
  In order to evaluate~\eqref{eq:crl1} we can repeatedly apply the variation identities~\eqref{eq:variation_identity} and~\eqref{eq:ddt_r}. This yields
  \[
    \E_\m  \prod_{j = 1}^m \int_{\Sigma_0}v_j\wedge dh = (-i)^m\int_{\Sigma^m}\det[1_{j\neq k} \Dd^{-1}_{\alpha_0 + \alpha_\m}(p_j,p_k) + 1_{j=k} r_{\alpha_0 + \alpha_\m}(p_j)]\wedge \prod_{j = 1}^m \beta_j(p_j).
  \]
  Note that we can replace $\beta_j$ with $2i\Im\beta_j$ as the wedge product of two $(1,0)$ forms is zero by convention. This implies~\eqref{eq:determinantal_identity}.
\end{proof}

We have the following immediate corollary:

\begin{cor}
  \label{cor:centered}
  Let $h$, $\E_\m $ and $v_1,\dots, v_m$ be as in Theorem~\ref{thmas:determinantal_identity}. Then
  \[
    \E_\m  \prod_{j = 1}^m \int_{\Sigma_0}v_j\wedge d(h - \E_\m h) = \int_{\Sigma^m}\det[1_{j\neq k} \Dd^{-1}_{\alpha_0 + \alpha_\m}(p_j,p_k)]\wedge \prod_{j = 1}^m v_j(p_j).
  \]
\end{cor}

\begin{cor}
  \label{cor:correlations}
  Let $h = \phi + \int\psi$ and $\E_\m $ be such as in Theorem~\ref{thmas:determinantal_identity}. Let $p_1,q_1,\dots, p_m,q_m\in \Sigma_0$ be distinct points and $\gamma_1,\dots,\gamma_m\subset \Sigma_0$ be disjoint paths such that $\gamma_j$ connects $p_j$ and $q_j$ and oriented towards $p_j$. Write formally
  \[
    h(p_j) - h(q_j) = \phi(p_j) - \phi(q_j) + \int_{\gamma_j}\psi,\qquad h_0 = h - \E_\m h.
  \]
  Put $\gamma_j^+ = \gamma_j$ and $\gamma_j^- = \sigma(\overleftarrow{\gamma_j})$ where the arrow indicates that the orientation is reversed. Then we have
  \[
    \begin{split}
      &\E_\m \prod_{j = 1}^m(h(p_j) - h(q_j)) = (-1)^m\sum_{s\in \{ \pm \}^m}\int_{\gamma_1^{s_1}}\dots\int_{\gamma_m^{s_m}}\det[1_{j\neq k} \Dd^{-1}_{\alpha_0 + \alpha_\m}(p_j,p_k) + 1_{j=k} r_{\alpha_0 + \alpha_\m}(p_j)],\\
      &\E_\m \prod_{j = 1}^m(h_0(p_j) - h_0(q_j)) = (-1)^m\sum_{s\in \{ \pm \}^m}\int_{\gamma_1^{s_1}}\dots\int_{\gamma_m^{s_m}}\det[1_{j\neq k} \Dd^{-1}_{\alpha_0 + \alpha_\m}(p_j,p_k)].
    \end{split}
  \]
\end{cor}
\begin{proof}
  It is possible to construct a sequence of smooth 1-forms $(v_{\gamma_1}^i,\dots,v_{\gamma_m}^i)_{i>0}$ such that 
  \begin{itemize}
    \item For each $j = 1,\dots, m$ the supports of $v_{\gamma_j}^1,v_{\gamma_j}^2,\dots$ concentrate around $\gamma_j$.
    \item For each smooth 1-form $v$ and $j = 1,\dots, m$ we have $\lim\limits_{i\to \infty} \int_{\Sigma_0} v_{\gamma_j}^i\wedge v = \int_{\gamma_j}v$.
    \item We have $\lim\limits_{i\to\infty}\E_\m \prod_{j = 1}^m\int_{\Sigma_0}v_{\gamma_j}^i\wedge dh = \E_\m \prod_{j = 1}^m(h(p_j) - h(q_j))$.
  \end{itemize}
  Once we constructed this sequence we can extend each $v_{\gamma_j}^i$ to $\Sigma$ such that $\sigma^\ast v_{\gamma_j}^i = v_{\gamma_j}^i$ and apply Theorem~\ref{thmas:determinantal_identity} and Corollary~\ref{cor:centered} and to get the result. Note that if $v$ is a smooth 1-form on $\Sigma$ then
  \[
    \int_{\Sigma_0^\op} v_{\gamma_j}^i\wedge v = -\int_{\Sigma_0} v_{\gamma_j}^i\wedge \sigma^\ast v \xrightarrow[i\to \infty]{} -\int_{\gamma_j} \sigma^\ast v = \int_{\sigma(\overleftarrow{\gamma_j})}v
  \]
  which explains the orientation reversal in the definition of $\gamma_j^-$.

  Let us construct a sequence $v_\gamma^i$ for a given smooth path $\gamma$. We can break path into small segments and construct a sequence for each of them, thus we can work in a coordinate chart, that is, we can assume that $\gamma$ is a path in $\CC$ and work there. By the implicit function theorem one can find a smooth real-valued function $f$ such that $\nabla f\neq 0$ in a neighborhood of $\gamma$ and $\gamma\subset \{z\colon f(z) = 0\}$. Since $\nabla f$ evaluated at a point on $\gamma$ is a normal vector to $\gamma$ it is possible to find a measure $\mu$ supported on $\gamma$ such that for each 1-form $v = v_1\,dx + v_2\,dy$ we have
  \[
    \int_\CC (f_xv_2 - f_y v_1)\,d\mu = \int_\gamma v.
  \]
  Let $\rho$ be a smooth non-negative function with a compact support such that $\int_\CC\rho\,dx\wedge dy = 1$. Define
  \[
    \rho_k(z) = k^2\rho(kz),\qquad H_k(z) = \int_\gamma \rho_k(z-w)\,d\mu(w)
  \]
  and put
  \[
    v_\gamma^k = H_k\,df.
  \]
  It is easy to verify that $v_\gamma^k$ satisfy the desired properties.
\end{proof}

\section{Proof of the main result}
\label{sec:proof}

Before we dive into the details, let us give a quick summary of the proof. We will prove only Theorem~\ref{thma:surfaces} as Theorem~\ref{thma:domains} is a particular case of the former. Given $p\in \Sigma_0$ define $p^{[+]} = p$ and $p^{[-]} = \sigma(p)$. First, we use an abstract argument to find a gauge relating $\sqrt{\omega^{[\pm]}(p^{[\pm]})}f^{[\pm,\pm]}(p^{[\pm]},q^{[\pm]})\sqrt{\omega^{[\pm]}(q^{[\pm]})}$ with the Cauchy kernel of a line bundle $\Ff_0\otimes \Ll_{\alpha_0}$ on the double $\Sigma$, where $\alpha_0$ is an unknown anti-holomorphic differential. Then we use Theorem~\ref{thmas:determinantal_identity} and its corollaries to deduce that the determinantal identities involving $f^{[\pm,\pm]}$ define correlations of the compactified free field with the shift $\psi_0 = 2\pi^{-1}\Im\alpha_0$.

Let us now dive into details. In what follows we will use the notation introduced in Section~\ref{subsec:Main results: bordered Riemann surfaces} and the assumptions made in Theorem~\ref{thma:surfaces}. Without loss of generality we can assume that $\omega$ does not vanish on $\Sigma_0\smm\{ x_1,\dots,x_N \}$ (otherwise we can just add its zeros to $\{ x_1,\dots, x_N \}$). Recall that $\Sigma$ denotes the double of $\Sigma_0$ and $\sigma:\Sigma\to \Sigma$ is the associated involution (see Section~\ref{subsec:Schottky double}). Let us put $x_{N+j} = \sigma(x_j)$. Given $p,q\in \Sigma\smm(\partial\Sigma_0\cup \{ x_1,\dots, x_{2N} \}),\ p\neq q$, define
\begin{equation}
  \label{eq:f_and_omega}
  f(p,q) = \begin{cases}
    f^{[+,+]}(p,q),\quad p,q\in \Sigma_0,\\
    -f^{[+,-]}(p,\sigma(q)),\quad p\in \Sigma_0, q\in \sigma(\Sigma_0),\\
    f^{[-,+]}(\sigma(p),q),\quad p\in \sigma(\Sigma_0), q\in \Sigma_0,\\
    -f^{[-,-]}(\sigma(p),\sigma(q)),\quad p,q\in \sigma(\Sigma_0),
  \end{cases}
  \qquad \omega(p) = \begin{cases}
    \omega^{[+]}(p),\quad p\in \Sigma_0,\\
    \omega^{[-]}(\sigma(p)),\quad p\in \sigma(\Sigma_0).
  \end{cases}
\end{equation}
Then $f(p,q)$ is a holomorphic multivalued function in both variables and $\omega$ is a holomorphic multivalued differential. Note that
\begin{equation}
  \label{eq:sigmaf}
  f(\sigma(p),\sigma(q)) = -\overline{f(p,q)}
\end{equation}
and, by the assumption on $U_2$, when $p,q$ are close we have
\begin{equation}
  \label{eq:f_near_diag}
  f(p,q)f(q,p)\omega(p)\omega(q) = \frac{dz(p)dz(q)}{4\pi^2(z(p)- z(q))^2} + O(1)
\end{equation}
where $z$ is any local coordinate.

\begin{lemma}
  \label{lemma:extension}
  Let $p_1,p_2\in \Sigma\smm(\partial\Sigma_0\cup \{ x_1,\dots, x_{2N} \})$ be two distinct points. Then $f(p_1,p)f(p,p_2)\omega(p)$ extends to $\Sigma\smm\{ p_1,p_2 \}$ as a single valued holomorphic differential with simple poles at $p_1,p_2$.
\end{lemma}
\begin{proof}
  Since both $U_2,U_3$ have Dirichlet boundary conditions in each variable, we can extend them to $\Sigma$ as harmonic differentials  such that $U_m(\ldots, \sigma(p_j),\ldots) = -U_m(\ldots, p_j, \ldots)$ for each $m = 2,3$ and $j = 1,\dots, m$, where the rest of the arguments are fixed. Define
  \[
    \Aa_2 = U_2^{(1,0), (1,0)}, \qquad \Aa_3 = U_3^{(1,0),(1,0),(1,0)},
  \]
  that is, $\Aa_m$ is the $(1,0)$ component of $U_m$. Note that $\Aa_2$ and $\Aa_3$ are meromorphic bi- and tri-differentials respectively. The bi-differential $\Aa_2$ is defined on $\Sigma^2\smm\diag$ with and has a quadratic pole along the diagonal, while $\Aa_3$ is well-defined and holomorphic on $\Sigma^3$ by assumptions from Theorem~\ref{thma:surfaces}. On the other hand, the definitions of $f$ and $\omega$ imply that
  \[
    \begin{split}
      \Aa_2(p_1,p_2) &= f(p_1,p_2)f(p_2,p_1)\cdot \omega(p_1)\omega(p_2),\\
      \Aa_3(p_1,p_2,p_3) &= \Big[f(p_1,p_2)f(p_2,p_3)f(p_3,p_1) + f(p_1,p_3)f(p_3,p_2)f(p_2,p_1)\Big]\cdot\omega(p_1)\omega(p_2)\omega(p_3).
    \end{split}
  \]
  Define
  \[
    \wtd\Aa_3(p_1,p_2,p_3) = \Big[f(p_1,p_2)f(p_2,p_3)f(p_3,p_1) - f(p_1,p_3)f(p_3,p_2)f(p_2,p_1)\Big]\cdot\omega(p_1)\omega(p_2)\omega(p_3).
  \]
  A priori $\wtd\Aa_3$ is a tri-differential defined on $(\Sigma\smm(\partial\Sigma_0\cup\{ x_1,\dots,x_{2N} \}))^3\smm\diag$ with simple poles along the diagonals. However, we can write
  \begin{equation}
    \label{eq:exe1}
    \wtd\Aa_3(p_1,p_2,p_3)^2 = \Aa_3(p_1,p_2,p_3)^2 - 4\Aa_2(p_1,p_2)\Aa(p_2,p_3)\Aa(p_1,p_3)
  \end{equation}
  which implies that $\wtd\Aa_3^2$ extends to the whole $\Sigma^3\smm\diag$. Let us argue that $\wtd\Aa_3$ itself extends as well. To this end it is enough to show that the zero locus of the right-hand side of~\eqref{eq:exe1} has even multiplicity. Let $\Zz_\odd\subset \Sigma^3$ be the union of those components of this locus that have odd multiplicity. As the left-hand side of~\eqref{eq:exe1} is a full square on $(\Sigma\smm(\partial\Sigma_0\cup\{ x_1,\dots,x_{2N} \}))^3$ we have 
  \[
    \Zz_\odd \cap (\Sigma\smm(\partial\Sigma_0\cup\{ x_1,\dots,x_{2N} \}))^3 = \varnothing.
  \]
  Since $\Zz_\odd$ is a holomorphic subvariety, this can only happen if $\Zz_\odd$ is either empty or consists of components of the form $\{ p_0 \}\times \Sigma\times \Sigma,\ p_0\in \partial\Sigma_0\cup\{ x_1,\dots,x_{2N} \}$. In the latter case we would have
  \[
    \Aa_3(p_0,p_2,p_3)^2 - 4\Aa_2(p_0,p_2)\Aa(p_2,p_3)\Aa(p_0,p_3) = 0
  \]
  for any $p_2,p_3\in \Sigma$. However, if we fix $p_3$ and consider $p_2$ close to $p_0$, then the principal part of the near-diagonal asymptotics of $\Aa_2(p_0,p_2)$ will dominate everything else on the right-hand side and the expression above will be non-zero which is a contradiction.

  We conclude that
  \[
    f(p_1,p_2)f(p_2,p_3)f(p_3,p_1)\cdot\omega(p_1)\omega(p_2)\omega(p_3) = \tfrac12 (\Aa_3(p_1,p_2,p_3) + \wtd\Aa_3(p_1,p_2,p_3))
  \]
  extends as a tri-differential on $\Sigma^3\smm\diag$ with simple poles along the diagonals which implies the result.
\end{proof}

\begin{cor}
  \label{cor:f++residue}
  There exists a global constant $c = \pm1$ such that for any $q\in \Sigma_0\smm\{ x_1,\dots, x_N \}$ and a local coordinate $z$ defined in a neighborhood of $q$ we have
  \[
    f^{[+,+]}(p,q)\omega(q) = \frac{cdz(q)}{2\pi i(z(q) - z(p))} + O(1)
  \]
  as $p\to q$.
\end{cor}
\begin{proof}
  Lemma~\ref{lemma:extension} implies that $f(p,q)$ does not have an essential singularity at $q$. The statement follows from properties of $U_2$ assumed in Theorem~\ref{thma:surfaces}.
\end{proof}

Without loss of generality we can assume that $c = -1$ in Corollary~\ref{cor:f++residue}. Indeed, replacing $f^{[\pm,\pm]}$ with $-f^{[\pm,\pm]}$ is equivalent to replacing $h$ with $-h$ and $\psi_0$ with $-\psi_0$ in Theorem~\ref{thma:surfaces}. Thus, we have
\begin{equation}
  \label{eq:f++residue}
  f^{[+,+]}(p,q)\omega(q) = -\frac{dz(q)}{2\pi i(z(q) - z(p))} + O(1),\qquad p\to q.
\end{equation}

Let us now fix a simplicial basis $A_1,\dots, A_g,B_1,\dots, B_g\in H_1(\Sigma,\ZZ)$ as in Section~\ref{subsec:Schottky double} and denote by $\Ff_0$ be the spin line bundle with zero characteristics defined as in Section~\ref{subsec:Family of Cauchy--Riemann operators}. Further, let us fix an open cover $(U_i)_{i\in I_x\cup I_0\cup I_\partial}$ of $\Sigma$ described as follows. The index sets are
\[
  \begin{split}
    &I_x = I^{[+]}_x\cup I^{[-]}_x,\qquad I_x^{[+]} = \{ 1,\dots, N \},\quad I_x^{[-]} = \{ N+1,\dots, 2N \}\\
    &I_0 = I^{[+]}_0\cup I^{[-]}_0,\qquad I_0^{[+]} = 2N + \{ 1,\dots, k \},\quad I_0^{[-]} = 2N + \{ k+1,\dots, 2k \}\\
    &I_\partial = I^{[+]}_\partial\cup I^{[-]}_\partial,\qquad I_\partial^{[+]} = (2N+2k+1)+\{0,\dots, n \},\quad I_\partial^{[-]} = (2N+2k+1)+\{ n+1,\dots, 2n+1\}.
  \end{split}
\]
where $k$ is a large enough integer. By abusing the notation we denote by $\sigma$ the involution on $I_x\cup I_0\cup I_\partial$ defined such that
\[
  \begin{split}
    &\sigma(i) = i+N,\quad i\in I^{[+]}_x,\\
    &\sigma(i) = i+k,\quad i\in I^{[+]}_0,\\
    &\sigma(i) = i+n+1,\quad i\in I_\partial^{[+]}.
  \end{split}
\]
We assume that the sets $U_i$ from the cover satisfy the following properties:
\begin{itemize}
  \item For each $i$ we have $\sigma(U_i) = U_{\sigma(i)}$.
  \item For each $i\in I_x^{[+]}$ the set $U_i$ is a simply connected neighborhood of $x_i$ not intersecting $\partial\Sigma_0$ or the neighborhoods of other $x_j$-s.
  \item For each $i\in I_0^{[+]}$ the set $U_i\subset \Sigma_0$ is simply connected and does not intersect $\{ x_1,\dots,x_{2N} \}\cup\partial \Sigma_0$.
  \item For each $i$ such that $2N+2k+1+i\in I_\partial^{[+]}$ the set $U_{2N+2k+1+i}$ is a neighborhood of the boundary component $B_i$ which does not intersect neighborhoods of $x_1,\dots, x_{2N}$ or other components $B_j$-s for $j\neq i$.
\end{itemize}
For each $j$ fix a holomoprhic section $\zeta_j\in H^0(U_j, \Ff_0)$ with no zeros (this is possible for each $U_j$ because all holomorphic line bundles on a non-compact Riemann surface are trivial). Let $\psi_{ij}:U_i\cap U_j\to \CC^\ast$ be holomorphic functions such that for each $i\neq j$
\[
  \zeta_j = \psi_{ij}\zeta_i,
\]
that is, $\psi_{ij}$ are transition functions representing $\Ff_0$ if we trivialize it over each $U_i$ such that $\zeta_i$ is sent to 1. By adjusting the choice of $\zeta_j$ we can assure that for every $i\neq j$
\[
  \psi_{ij}(p) = \overline{\psi_{\sigma(i)\sigma(j)}(\sigma(p))}.
\]
Fix $p_0\in\Sigma_0\smm(\partial\Sigma_0\cup\{ x_1,\dots, x_N \})$. For each $i\in I_0^{[+]}$ choose single valued branches $f_i(p,p_0)$ of $f(p,p_0)$ and $\omega_i$ of $\omega$. Further, for every $i\in I_0^{[+]}$ choose $h_i$ such that
\[
  \omega_i = h_i^2\zeta_i^2.
\]
Finally, define $H_i$ as follows:
\[
  \begin{split}
    &H_i(p) = f_i(p,p_0)h_i(p),\qquad i\in I_0^{[+]},\\
    &H_i(p) = \overline{H_{\sigma(i)}(\sigma(p))},\qquad i\in I_0^{[-]},\\
    &H_i(p) = 1,\qquad i\in I_x\cup I_\partial,\\
  \end{split}
\]
Consider the line bundle $\Ll$ on $\Sigma$ represented by the transition functions $\vphi_{ij},\ i,j\in I_x\cup I_0\cup I_\partial,$ defined such that:
\[
  \begin{split}
    &\vphi_{ij}(p) = \psi_{ij}(p)^{-1}\frac{H_i(p)}{H_j(p)},\qquad (i,j)\notin I_\partial\times I_\partial,\\
    &\vphi_{ij}(p) = \psi_{ij}(p)^{-1}\frac{f(p,\sigma(p_0))}{f(p,p_0)},\qquad i\in I_\partial^{[+]}\ j\in I_\partial^{[-]},\\
  \end{split}
\]
Using Lemma~\ref{lemma:extension} it is easy to see that $\vphi_{ij}$ are well-defined as meromorphic functions. Moreover, arguing as in the proof of Lemma~\ref{lemma:extension} one can show that $\vphi_{ij}$ has no zeros and no poles provided that $p_0$ is chosen generically enough and $U_j$ are chosen properly. Note also that we have
\begin{equation}
  \label{eq:symmetric}
  \vphi_{ij}(p) = \overline{\vphi_{\sigma(i)\sigma(j)}(\sigma(p))}.
\end{equation}
Let us define
\[
  \begin{split}
    &C_{[++]}(p,q) = f(p,q)\frac{f(q,p_0)\omega(q)}{f(p,p_0)},\\
    &C_{[+-]}(p,q) = f(p,q)\frac{f(q,\sigma(p_0))\omega(q)}{f(p,p_0)},\\
    &C_{[-+]}(p,q) = f(p,q)\frac{f(q,p_0)\omega(q)}{f(p,\sigma(p_0))},\\
    &C_{[--]}(p,q) = f(p,q)\frac{f(q,\sigma(p_0))\omega(q)}{f(p,\sigma(p_0))}.
  \end{split}
\]
By Lemma~\ref{lemma:extension} every $C_{[\pm\pm]}$ can be extended to $\Sigma\times\Sigma$ as a meromorphic function in the first variable and a meromorphic differential in the second. Adjusting $p_0$ and the cover again we can make sure that for each fixed $q$ the function $C_{[\pm\pm]}(p,q)$ does not have poles on $U_i\smm\{ q \},\ i\in I_x\cup I_\partial$.

For each $i$ denote by $\zeta_i^{-1}$ the section of $\Ff_0^\ast$ over $U_i$ such that $\zeta_i\cdot \zeta_i^{-1} = 1$. Finally, for $p\in U_i$ and $q\in U_j$ let us define:
\[
  D_{ij}^{-1}(p,q) = -\zeta_i(p)H_i(p)C_{[s_1s_2]}(p,q)H_j(q)^{-1}\zeta_j^{-1}(q),\qquad i\in I^{[s_1]},\ j\in I^{[s_2]},\qquad s_1,s_2\in \{ \pm \}
\]
where $I^{[s]} = I_x^{[s]}\cup I_0^{[s]}\cup I_\partial^{[s]}$.

\begin{lemma}
  \label{lemma:Ddij_is_Cauchy_kernel}
  There is a global constant $c = \pm1$ such that the family $(c\Dd_{ij}^{-1})_{i,j}$ form together the Cauchy kernel of the operator $\dbar$ acting on smooth sections of $\Ff_0\otimes \Ll$.
\end{lemma}
\begin{proof}
  We need to show that $c\Dd_{ij}^{-1}$ put together form an object $\Dd^{-1}$ which 
  \begin{itemize}
    \item is a meromorphic section of $\Ff_0\otimes \Ll$ in the first variable,
    \item is a meromorphic section of $K_\Sigma\otimes (\Ff_0\otimes \Ll)^\ast$ in the second variable,
    \item has a simple pole along the diagonal with the residue $\frac{1}{2\pi i}$ and no other poles,
  \end{itemize}
  see Section~\ref{subsec:Cauchy kernel of a line bundle} for more details. The verification of the first two properties as well as of the fact that $\Dd^{-1}$ has no poles outside the diagonal is absolutely straightforward (although maybe tedious). The residue calculation follows from Corollary~\ref{cor:f++residue} and the remark after it.
\end{proof}

\begin{cor}
  \label{cor:Ll_is_Llalpha}
  There exist harmonic differentials $\psi_0$ and $\m$ on $\Sigma$ such that $\sigma^\ast\psi_0 = -\psi_0$, $\sigma^\ast\m = \m$, $\m$ has integer periods and $\Ll\cong\Ll_{\alpha_0+\alpha_\m}$ where
  \[
    \alpha_0 = \pi i\psi_0^{0,1},\qquad \alpha_\m = \pi i \m^{0,1}.
  \]
\end{cor}
\begin{proof}
  By Lemma~\ref{lemma:Dd-1_exists} and Lemma~\ref{lemma:Ddij_is_Cauchy_kernel} we have $\deg\Ll = 0$, thus $\Ll\cong \Ll_\alpha$ for some $\alpha$ by Lemma~\ref{lemma:flat_line_bundles}. Recall that the set of transition functions $\vphi_{ij}$ determining $\Ll$ is invariant under $\sigma$. This implies that $\alpha$ and $\sigma^\ast\bar\alpha$ determine the same point in $\Jac(\Sigma)$. Define
  \[
    \m = \pi^{-1}\Im(\alpha - \sigma^\ast\bar\alpha),\qquad \psi_0 = \pi^{-1}\Im(\alpha + \sigma^\ast\bar\alpha)
  \]
  By~\eqref{eq:Jac_torus} $\m$ must be integer and, by construction, $\sigma^\ast\psi_0 = -\psi_0$, $\sigma^\ast\m = \m$ and $\alpha = \alpha_0 + \alpha_\m$.
\end{proof}

\begin{proof}[Proof of Theorem~\ref{thma:domains} and Theorem~\ref{thma:surfaces}]
  We will only prove Theorem~\ref{thma:surfaces} as Theorem~\ref{thma:domains} is a particular case of it. Let $\psi_0$ and $\m$ be as in Corollary~\ref{cor:Ll_is_Llalpha}, let $h = \phi + \int \psi$ denotes the compactified free field with the shift $\psi_0$. By the Lemma~\ref{lemma:Em_exists} and Lemma~\ref{lemma:Ddij_is_Cauchy_kernel} we have $\EE\exp(\pi i Q_\m(\psi - \psi_0))\neq 0$ and the expectation $\E_\m$ is well-defined. Let $p_1,\dots, p_m\in \Sigma_0$ be distinct points and assume that for each $j$ the index $i_j$ is such that $p_j\in U_{i_j}$. Using the definition of $f$ and $\Dd^{-1}_{ij}$ it is easy to see that
  \[
    \sum_{s\in \{ \pm \}^m}\det\left[1_{j\neq k}f^{[s_j,s_k]}(p_j,p_k)\right]\prod_{j = 1}^m\omega^{[s_j]}(p_j) = \sum_{s\in \{ \pm \}^m} \left(\prod_{j = 1}^m-s_j\right) \det\left[1_{j\neq k}\Dd_{i_ji_k}^{-1}(p^{[s_j]}_j, p^{[s_k]}_{i_k})\right]
  \]
  where $p^{[+]} = p$ and $p^{[-]} = \sigma(p)$. Using Lemma~\ref{lemma:Ddij_is_Cauchy_kernel} we can rewrite it more invariantly:
  \begin{equation}
    \label{eq:prf1}
    \sum_{s\in \{ \pm \}^m}\det\left[1_{j\neq k}f^{[s_j,s_k]}(p_j,p_k)\right]\prod_{j = 1}^m\omega^{[s_j]}(p_j) = \sum_{s\in \{ \pm \}^m} \left(\prod_{j = 1}^m-s_j\right) \det\left[1_{j\neq k}\Dd^{-1}(p^{[s_j]}_j, p^{[s_k]}_{i_k})\right]
  \end{equation}
  where $\Dd^{-1}$ is the Cauchy kernel of the $\dbar$ operator acting on smooth sections of $\Ff_0\otimes \Ll$. By Corollary~\ref{cor:Ll_is_Llalpha} we have $\Ll\cong\Ll_{\alpha_0+\alpha_\m}$, hence $\Dd^{-1}$ is intertwined with the kernel $\Dd^{-1}_{\alpha_0+\alpha_\m}$ (see Section~\ref{subsec:Family of Cauchy--Riemann operators}), so we have
  \begin{equation}
    \label{eq:prf2}
    \sum_{s\in \{ \pm \}^m} \left(\prod_{j = 1}^ms_j\right) \det\left[1_{j\neq k}\Dd^{-1}(p^{[s_j]}_j, p^{[s_k]}_{i_k})\right] = \sum_{s\in \{ \pm \}^m} \left(\prod_{j = 1}^ms_j\right) \det\left[1_{j\neq k}\Dd^{-1}_{\alpha_0+\alpha_\m}(p^{[s_j]}_j, p^{[s_k]}_{i_k})\right].
  \end{equation}
  Combining~\eqref{eq:prf1},~\eqref{eq:prf2} and Corollary~\ref{cor:correlations} we conclude that
  \[
    \sum_{s\in \{ \pm \}^m}\det\left[1_{j\neq k}f^{[s_j,s_k]}(p_j,p_k)\right]\prod_{j = 1}^m\omega^{[s_j]}(p_j) = \E_\m\prod_{j = 1}^m d(h - \E_\m h)(p_j)
  \]
  where $h= \phi + \int \psi$ is the compactified free field such that $\psi_0 = 2\pi^{-1}\Im\alpha_0$ is the shift making $\psi$ integer and used to define $\E_\m$.
\end{proof}

We finish the section with the following technical remark:

\begin{rem}
  \label{rem:multivalued_spinors}
  As follows from Remark~\ref{rem:multivalued_functions}, meromorphic sections of $\Ff_0\otimes \Ll_{\alpha_0 + \alpha_\m}$ are in one-to-one correspondence with meromorphic multivalued sections of $\Ff_0$ with multiplicative monodromy given by the character
  \[
    \chi_0(\gamma)\chi_\m(\gamma) = \exp(\pi i\int_\gamma(\psi_0 + \m))
  \]
  representing the same point in $\Jac(\Sigma)$ as $\alpha_0 + \alpha_\m$. In particular, $\Dd_{\alpha_0 + \alpha_\m}^{-1}(p,q)$ is gauge equivalent to a meromorphic multivalued section of $\Ff_0\boxtimes\Ff_0$ with the monodromy given by $\chi_0\chi_\m$ in the first variable and by $(\chi_0\chi_\m)^{-1}$ in the second. This provides us with a useful tool to reconstruct the height shift $\psi_0$ and the differential $\m$: it is enough to find a gauge relating the above function $f(p,q)$ (or, more generally, the spinor $\sqrt{\omega(p)}f(p,q)\sqrt{\omega(q)}$) with such multivalued section, then the monodromy determines $\psi_0$ and $\m$ as anti-symmetric and summetric components of $\psi_0 + \m$ up to an integer differential with even periods (cf. Remark~\ref{rem:dependence_on_psi0}). As we will see below, it allows to easily determine the height shift in various special cases.
\end{rem}

\section{Examples}

\label{sec:Examples}

In many cases of interest functions $f^{[\pm,\pm]}$ satisfy explicit boundary conditions that allow to extend them to the double $\Sigma$ and to determine the line bundle $\Ff_0\otimes \Ll_{\alpha_0}$ explicitly. Below we consider a number of examples which arise from Temperleyan boundary conditions and their simple modifications. We begin with the following general remark. Assume that $\Omega_k$ is a sequence of discrete domains approximating a domain $\Omega$, and assume that there is a sequence $\partial_k$ of boundary arcs of $\Omega_k$ converging to an arc $\partial\subset \partial \Omega$ such that $\Omega_k$ has Temperleyan boundary along $B_k$. Assume that $f^{[\pm,\pm]}$ in $\Omega$ arise as limits of $F^{[\pm,\pm]}_k$ in $\Omega_k$. Then one can set the conventions determining $F^{[\pm,\pm]}_k$ such that
\begin{equation}
  \label{eq:temp_bc}
  \begin{split}
    &f^{[+,+]}(p,q) - f^{[-,+]}(p,q) = 0,\qquad p\in \partial,\ q\in \Omega,\\
    &(f^{[+,+]}(p,q) - f^{[+,-]}(p,q))dq = 0 \quad \text{along $\partial$ for each $p\in \Omega$ fixed}.
  \end{split}
\end{equation}
In particular, if we define $f(p,q)$ on the double of $\Omega$ by~\eqref{eq:f_and_omega}, then $f(p,q)\,dq$ extends holomorphically to $\partial\times \partial\smm\diag$ provided that $f^{[\pm, \pm]}(p,q)\,dq$ stays bounded if $p,q$ approach $\partial$ and stay on a definite distance from each other.

\subsection{Cylinder with Temperleyan boundaries}
\label{subsec:Cylinder with Temperley boundaries}

Fix $\tau>0$ and consider the discrete cylinder
\[
  \Cc_k^{DD} = \{ (x,y)\in \ZZ^2\ \mid\ 0\leq y\leq 2\lfloor \tfrac{k\tau}{2} \rfloor \}/_{(x,y)\sim (x+2k,y)}.
\]
The scaling limit of $\Cc_k^{DD}$ is the cylinder
\begin{equation}
  \label{eq:Cc_tau}
  \Cc_\tau = \{ x+iy\ \mid\ 0\leq y\leq \tfrac12\tau \}/_{z\sim z+1}.
\end{equation}
The dimer model sampled on $\Cc_k$ with uniform weights then has Temperleyan boundary conditions on both boundary components of the underlying surface. To construct the Kasteleyn operator on $\Cc_k$ enumerating the dimer configurations one can start with the standard Kasteleyn weights on $\ZZ^2$ and then consider the corresponding Kasteleyn operator acting on multivalued functions with multiplicative monodromy $-1$.

Consider the double $\Tt$ of $\Cc_\tau$, that is,
\begin{equation}
  \label{eq:Tt_tau}
  \Tt = \CC/(\ZZ + i\tau\ZZ).
\end{equation}
The cotangent bundle $K_\Tt$ is trivial, hence the spin line bundle $\Ff_0$ has degree zero and hence its sections can be identified with multivalued functions on $\Tt$. Since $\Ff_0^{\otimes 2}$ is trivial, the monodromies must be $\pm1$. To compute them, recall the definition~\eqref{eq:def_of_q} of the quadratic form $q$ associated with a spin bundle and the definition~\eqref{eq:def_of_q0} of the particular form $q_0$ that defines $\Ff_0$. Combining these definitions it is easy to see that sections of $\Ff_0$ can be identified with multivalued functions on the torus $\Tt$ with monodromies $-1$ along both basic cycles.

At the same time,~\eqref{eq:temp_bc} (and the standard regularity estimates on $f^{[\pm,\pm]}$ near the boundary) shows that the function $f(p,q)$ defined in the proof of Theorem~\ref{thma:surfaces} (see~\eqref{eq:f_and_omega}) extends to $\Tt$ in both variables and has monodromy only around the $B$-cycle (that is, around the cylinder), but not around the $A$-cycle (the path connecting two boundary components of $\Cc_\tau$ and its image under the involution). Combined with the remark above this implies that $f(p,q)$ (say, when $q$ is fixed and $p$ varies) is a multivalued section of $\Ff_0$ with the monodromy $-1$ along the $A$-cycle. Combining this with Remark~\ref{rem:multivalued_spinors} we conclude that the scaling limit $h$ of dimer height functions on $\Cc_k^{DD}$ has half-integer height change between two boundary components.

We can also consider a slightly modified version of $\Cc_k^{DD}$ given by
\[
  \Cc_k^{ND} = \{ (x,y)\in \ZZ^2\ \mid\ 1\leq y\leq 2\lfloor \tfrac{k\tau}{2} \rfloor \}/_{(x,y)\sim (x+2k,y)},
\]
that is, $\Cc_k^{ND}$ is $\Cc_k^{DD}$ with the bottom layer removed. In this case both relations in~\eqref{eq:temp_bc} will appear with different sign when $\partial$ is the bottom component. This results in $f(p,q)$ having monodromies $-1$ along both basis cycles, which identifies it as a single valued section of $\Ff_0$. In particular, Remark~\ref{rem:multivalued_spinors} now tells us that the instanton component $\psi$ of $h$ is integer in the limit, that is, the shift $\psi_0 = 0$.

Let now $\Cc_k$ be either of $\Cc_k^{DD}$ and $\Cc_k^{ND}$. Let $b_k, w_k$ be two its boundary vertices colored in black and white respectively, and put $\Cc_k' = \Cc_k\smm\{ b_k,w_k \}$. Such a setup naturally appears in the study of the double-dimer interface. Assume that $b_k$ and $w_k$ converge to the boundary points $p_1,p_2\in \partial\Cc_\tau$ respectively. This results in asymptotic relations $f^{[+,+]}(p,q) = O(p-p_1)$ at $p_1$ and $f^{[+,+]}(p,q)= O((p-p_2)^{-1})$ at $p_2$, and the opposite for the second argument $q$. Moreover, if $p_1$ and $p_2$ (and hence $b_k$ and $w_k$ for a large $k$) belong to different boundary components we will have to modify the Kasteleyn operator by introducing an additional monodromy $-1$ along the $B$-cycle (boundary of $\Cc_\tau$). It follows that (say, when $q$ is fixed) $f(p,q)$ corresponds to the Cauchy kernel of the line bundle
\[
  \Ff_0\otimes\Ll_{\alpha_0}\otimes\Ll_{\alpha_1}\otimes\Oo_\Sigma(p_2 - p_1)
\]
where $\Ff_0\otimes\Ll_{\alpha_0}$ is the line bundle corresponding to the limit taken on $\Cc_k$, $\Ll_{\alpha_1}$ takes into account the additional monodromy, and $\Oo_\Sigma(p_2 - p_1)$ incorporates singularities at $p_1$ and $p_2$. 

To calculate the height change in this case we need to find periods of $\Im\alpha_2$, where $\alpha_2$ is such that $\Oo_\Sigma(p_2 - p_1)\cong \Ll_{\alpha_2}$. A standard calculation using Abel--Jacobi map (see e.g.~\cite[Section~9.4]{basok2023dimers}) implies that
\[
  \int_A\Im\alpha_2 = \pi\Re(p_2 - p_1),\qquad \int_B \Im\alpha_2 = \pi\tau^{-1}\Im(p_2 - p_1) = \int_A\Im\alpha_1 \mod \pi\ZZ.
\]
where the last equality follows from the definition $\alpha_1$. We conclude that the shift of the compactified free field obtained as a limit of height functions of $\Cc_k'$ is given by
\[
  \pi^{-1}\int_A \Im\alpha_0 + \Re(p_2 - p_1).
\]

\subsection{Cylinder with black and white boundaries}
\label{subsec:Cylinder with black and white boundaries}

This setup has been studied in details in~\cite{ChelkakDeiman}. Let $\Omega\subset \RR/\ZZ$ be a doubly connected bounded domain, and let $\Omega_k\subset \ZZ^2/_{z\sim z + 2k}$ be a sequence of discrete doubly connected domains approximating $\Omega$ and such that
\begin{itemize}
  \item[--] The top boundary component of each $\Omega_k$ is Temperleyan, that is, all the corners are black and have the same type.
  \item[--] The bottom boundary component of each $\Omega_k$ is white Temperleyan, that is, all the corners are white and are of the same type.
\end{itemize}
Let $\tau>0$ be such that $\Omega$ is conformally equivalent to $\Cc_\tau$, let $\phi: \Cc_\tau\to \Omega$ be a conformal map chosen such that the bottom boundary component is sent to the bottom boundary component. Let $\Tt$ denote the double of $\Cc_\tau$ and $\sigma$ denote the involution. Define
\[
  U_1 = \Tt\smm \partial_{\mathrm{bottom}}\Cc_\tau,\qquad U_2 = \Tt\smm\partial_{\mathrm{top}}\Cc_\tau.
\]
Then $U_1\cap U_2 = \Cc_\tau\smm\partial\Cc_\tau\cup \sigma(\Cc_\tau\smm\partial\Cc_\tau)$. Call the first component $U_{12}^+$ and the second $U_{12}^-$, and define
\[
  \vphi_{12}^+(z) = \phi'(z),\qquad \vphi_{12}^-(z) = -\overline{\phi'(\bar z)}.
\]
Let $\Ll$ denote the line bundle defined by the transition function $\vphi_{12}$ equal to $\vphi_{12}^\pm$ on $U_{12}^\pm$.

Consider now the functions $f^{[\pm,\pm]}$ on $\Omega$ originating from the dimer model on the sequence $\Omega_k$. Note that the Kasteleyn operator on each $\Omega_k$ can be chosen by taking standard Kasteleyn weights on $\ZZ^2$ twisting it with the monodromy $-1$ around the cylinder. Consequently, $f^{[\pm,\pm]}$ will inherit this monodromy $-1$. Summing this with~\eqref{eq:temp_bc} (see~\cite{ChelkakDeiman} for details) we can find a gauge relating $f^{[\pm,\pm]}$ with the Cauchy kernel in $\Ff_0\otimes \Ll$; thus, the height shift is determined by the point in the Jacobian corresponding to $\Ll$.

\subsection{Multiply connected domains: Kenyon's original setup}
\label{subsec:Multiply connected domains: Kenyon's original setup}

Let now $\Omega$ be a multiply connected domain with boundary components $B_0,B_1,\dots, B_n,\ n\geq 2,$ where $B_0$ denotes the outer component. Let $\Omega_k\subset \tfrac1k\ZZ^2$ be discrete domains with Temperleyan boundary conditions approximating $\Omega$. For each $k$ and $i = 1,\dots, n$ pick a black vertex $b_k$ on the outer boundary of $\Omega_k$ and a white vertex $w_k^{(i)}$ on the $i$-th boundary component of $\Omega_k$. Define $\Omega'_k = \Omega_k\smm\{ b_k,w_k^{(1)},\dots, w_k^{(n)} \}$. The dimer model in such $\Omega'_k$ has been originally considered in the foundational work~\cite{KenyonConfInvOfDominoTilings} of Kenyon. Let us show how the result of~\cite{nicoletti2025temperleyan} describing the limit of the height function can be re-derived using the language we developed.

\begin{rem}
  \label{rem:removing_adding}
  Note that from the dimer model perspective removing a white vertex $w_k^{(i)}$ is the same as adding an ``exposed'' black vertex incident to $w_k^{(i)}$ and to no other vertices in $\Omega_k$. The latter convention was taken in~\cite{KenyonConfInvOfDominoTilings}.
\end{rem}

Assume that $b_k, w_k^{(1)},\dots, w_k^{(n)}$ converge to points $p_0,p_1,\dots, p_n\in \partial\Omega$ respectively. Consider the functions $f^{[\pm,\pm]}$ induced by the sequence of dimer graphs $\Omega_k'$. Note that $f^{[\pm,\pm]}$ satisfy the conditions~\eqref{eq:temp_bc} on $\Omega\smm\{ p_0,\dots, p_n \}$. Moreover, for any fixed $q\in \Omega$, the function $f^{[+,\pm]}(p,q)$ has a simple pole at each $p_1,\dots, p_n$, a simple zero at $p_0$.

Let us construct such $f^{[\pm,\pm]}$. Denote by $\Sigma$ the double of $\Omega$, note that the genus of $\Sigma$ is $n$. Consider the line bundle $\Ll = \Oo_\Sigma(-p_0+p_1+\ldots+p_n)$. By construction $\Ll$ has a meromorphic section $\zeta$ with the divisor $-p_0+p_1+\ldots+p_n$. Note that $\deg \Ll = n-1 = \deg \Ff_0$, thus there exists $\alpha$ such that
\[
  \Ll \cong\Ff_0\otimes \Ll_\alpha.
\]
This allows to identify $\zeta$ with a multivalued meromorphic section of $\Ff_0$ (see Remark~\ref{rem:multivalued_spinors}). Note that $\zeta^2$ is then a multivalued meromorphic differential.

Let now $\Dd^{-1}$ be the Cauchy kernel of $\Ll$, also interpreted as a multivalued section of $\Ff_0$ in both variables. Given $p,q\in \Omega$ define
\[
  f^{[+,+]}(p,q) = \zeta(p)^{-1}\Dd^{-1}(p,q)\zeta(q)\,(dq)^{-1},\qquad f^{[-,+]}(p,q) = \zeta(\sigma(p))^{-1}\Dd^{-1}(\sigma(p),q)\zeta(q)\,(dq)^{-1}.
\]
A straightforward calculation shows that $f^{[\pm,\pm]}$ satisfy all the required properties. One can also show that $\alpha$ can be chosen such that $\sigma^\ast\alpha = \bar\alpha$, that it, $\psi_0 = 2\pi^{-1}\Im\alpha$ is anti-symmetric. We conclude that the compactified free field $h$ arising as a limit of dimer height functions on $\Omega_k'$ has the shift $\psi_0$ representing the same point as the divisor $-p_0+p_1+\ldots+p_n$ in the Jacobian of the double $\Sigma$.

\subsection{Multiply connected domains: punctures in the bulk}
\label{subsec:Multiply connected domains: punctures in the bulk}

In this example we take the same domains $\Omega_k$ as in the previous case but instead of puncturing them at the boundary we remove $n-1$ white vertices $w_k^{(1)},\dots, w_k^{(n-1)}$ from the bulk of each $\Omega_k$. This setup corresponds to the one considered in~\cite{basok2023dimers} and~\cite{BerestyckiLaslierRayI}. Assume that $w_k^{(i)}$ converge to a point $p_i\in \Omega$ for each $i = 1,\dots, n-1$. The Kasteleyn weights on $\Omega_k$ can be chosen such that $f^{[\pm,\pm]}$ are multivalued functions in $\Omega\smm\{ p_1,\dots, p_{n-1} \}$ with monodromy $-1$ around each inner boundary component and each $p_1,\dots, p_{n-1}$. Since $\Omega_k$ are Temperleyan, $f^{[\pm,\pm]}$ satisfy~\eqref{eq:temp_bc}.

We can construct such $f^{[\pm,\pm]}$ by generalizing the construction from the previous example. Let $\Sigma$ be the double of $\Omega$ and $p_{n-1+i} = \sigma(p_i),\ i = 1,\dots, n-1$. Consider the line bundle $\Oo_\Sigma(p_1 + \ldots + p_{2n-2})$. Recall that $n$ is the genus of $\Sigma$, therefore $\deg K_\Sigma = 2n-2$ and we can write
\[
  \Oo_\Sigma(p_1 + \ldots + p_{2n-2}) \cong K_\Sigma\otimes \Ll_{\alpha_0}
\]
for some anti-holomorphic $\alpha_0$. By construction, $\Oo_\Sigma(p_1 + \ldots + p_{2n-2})$ has a holomorphic section $\omega$ with the divisor $p_1+\ldots+p_{2n-2}$. The isomorphism above allows to identify $\omega$ with a multivalued holomorphic differential on $\Sigma$ with the same divisor and monodromy $\chi(\gamma) = \exp(-2i\int_\gamma\Im\alpha_0)$. One can choose $\omega$ and $\alpha_0$ such that $\sigma^\ast\alpha_0 = \bar\alpha_0$ and $\sigma^\ast\omega = \bar\omega$. 

Let us fix simple paths $\gamma_1,\dots,\gamma_{n-1}$ such that each $i$ the path $\gamma_i$ connects $p_i$ with $\sigma(p_i)$ and $\sigma(\gamma_i) = \gamma_i$. Fix a simplicial basis $A_1,\dots, A_n, B_1,\dots, B_n$ in $H_1(\Sigma,\ZZ)$ in the usual way (see Section~\ref{subsec:Schottky double}, note that $B_i$-s are inner boundary components of $\Omega$). Let $\Ff_0$ be the spin line bundle with zero characteristics. By the definition of the quadratic form associated with a spin line bundle (see Section~\ref{subsec:Spin line bundles}) we can identify $\sqrt{\omega}$ with a multivalued section of $\Ff_0$ over $\Sigma\smm\{ p_1,\dots, p_{2n-2} \}$ whose monodromy around a loop $\gamma$ is given by
\begin{equation}
  \label{eq:monodromy_of_w}
  \chi(\gamma) = -\exp\left( \pi i (\wind(\gamma, \omega) - q_0(\gamma)) \right) = \pm\exp\left(-i\int_\gamma \Im \alpha_0 + \pi i \sum_{j = 1}^{n-1} \gamma\cdot \gamma_j\right)
\end{equation}
where $\wind$ is defined as in Section~\ref{subsec:Spin line bundles} and the choice of $\pm$ in front of the exponent depend on the choice of $\alpha_0$. The symmetries of $\omega$ and $\alpha_0$ with respect to the involution $\sigma$ force the sign to be $-1$ if $\gamma = B_i$ for some $i$. Adjusting $\alpha_0$ we can also achieve the sign to be $+1$ when going around $A_i$ for each $i$.

Assume now that the line bundle $\Ff_0\otimes \Ll_{\tfrac{-\alpha_0}{2}}$ has a Cauchy kernel $\Dd^{-1}$. Identify this kernel with a multivalued section of $\Ff_0\boxtimes\Ff_0$ as usual and define
\[
  f^{[+,+]}(p,q) = \sqrt{\omega}(p)^{-1}\Dd^{-1}(p,q)\sqrt{\omega}(q)\,(dq)^{-1},\qquad f^{[-,+]}(p,q) = \sqrt{\omega}(\sigma(p))^{-1}\Dd^{-1}(\sigma(p),q)\sqrt{\omega}(q)\,(dq)^{-1}.
\]
It is easy to see that $f^{[\pm,\pm]}$ defined in this way satisfy all the required properties. We conclude that the shift of the compactified free field $h$ is given by $-\pi^{-1}\Im\alpha_0$. This agrees with the results of~\cite{basok2023dimers}.

\begin{rem}
  \label{rem:meaning_of_alpha0}
  Let us make some comments about the choice of $\alpha_0$.

  \begin{itemize}
    \item This differential has a geometric meaning: indeed, consider the Riemannian metric given by $|\omega|^2$, then this is the unique (up to a multiplicative constant) locally flat metric on $\Sigma$ with conical singularities of angles $4\pi$ at $p_1,\dots, p_{2n-2}$. The character $\chi(\gamma) = \exp(2i\int_\gamma\Im\alpha_0)$ is nothing but the holonomy of this metric. This appearance of the holonomy is very consonant with the approach of~\cite{BerestyckiLaslierRayI} where the dimer height function is related with the winding of branches of random trees, see~\cite{basok2023dimers} for more details.
    \item It is worth noting that the construction above fixes the square root $[\frac12\alpha_0]\in \Jac(\Sigma)$ of $[\alpha_0]\in \Jac(\Sigma_0)$, that is, it fixes $\pi^{-1}\Im\alpha_0$ up to an integer differential with even periods. Note that this may force $\alpha_0$ to be non-zero even when the holonomy of the metric above it trivial. This, for example, happens in the case of a Temperleyan cylinder discussed in the first example: in this case $\pi^{-1}\Im\alpha_0$ is required to have an odd $A$-period which results in the half-integer shift of the height change.
  \end{itemize}
\end{rem}

\printbibliography

@article{KLRR,
  TITLE = {{Dimers and circle patterns}},
  AUTHOR = {Kenyon, Richard and Yeung Lam, Wai and Ramassamy, Sanjay and Russkikh, Marianna},
  JOURNAL = {{Annales Scientifiques de l'{\'E}cole Normale Sup{\'e}rieure}},
  PUBLISHER = {{Soci{\'e}t{\'e} math{\'e}matique de France}},
  YEAR = {2021}
}

@ARTICLE{CLR1,
  title={Dimer model and holomorphic functions on t-embeddings of planar graphs},
  author={Chelkak, Dmitry and Laslier, Beno{\^i}t and Russkikh, Marianna},
  journal={Proceedings of the London Mathematical Society},
  volume={126},
  number={5},
  pages={1656--1739},
  year={2023},
  publisher={Wiley Online Library}
}

@ARTICLE{CLR2,
   author = {Dmitry Chelkak and Beno{\^i}t Laslier and Marianna Russkikh},
    title = {Bipartite dimer model: perfect t-embeddings and Lorentz-minimal surfaces},
  journal = {ArXiv e-prints},
archivePrefix = "arXiv",
   eprint = {2109.06272},
 primaryClass = "math.PR",
     year = 2022,
    month = nov
}

@article{ChelkakRamassamy,
     author = {Dmitry Chelkak and Sanjay Ramassamy},
     title = {Fluctuations in the {Aztec} diamonds via a space-like maximal surface in {Minkowski} 3-space},
     journal = {Confluentes Mathematici},
     pages = {1--17},
     year = {2024},
     publisher = {Institut Camille Jordan},
     volume = {16}
}

@article{berggren2024perfect,
  title={Perfect t-embeddings of uniformly weighted Aztec diamonds and tower graphs},
  author={Berggren, Tomas and Nicoletti, Matthew and Russkikh, Marianna},
  journal={International Mathematics Research Notices},
  volume={2024},
  number={7},
  pages={5963--6007},
  year={2024},
  publisher={Oxford University Press}
}

@article{berggren2025perfect,
  title={Perfect t-embeddings and the octahedron equation of the two-periodic Aztec diamond},
  author={Berggren, Tomas and Russkikh, Marianna},
  journal={arXiv preprint arXiv:2508.06697},
  year={2025}
}

@article{berggren2024lozenge,
  title={Perfect t-embeddings and Lozenge Tilings},
  author={Berggren, Tomas and Nicoletti, Matthew and Russkikh, Marianna},
  journal={arXiv preprint arXiv:2408.05441},
  year={2024}
}

@article{keating2025perfect,
  title={Perfect t-Embeddings of Uniformly Weighted Generalized Tower Graphs},
  author={Keating, David and Vu, Hieu Trung},
  journal={arXiv preprint arXiv:2509.18791},
  year={2025}
}

@article{KenyonConfInvOfDominoTilings,
    author = {Richard Kenyon},
    title = {Conformal invariance of domino tiling},
    journal = {Ann. Probab.},
    volume = {28},
    number = {2},
    pages = {759--795},
    year = {2000}
}

@article{KenyonGFF,
    author = {Richard Kenyon},
    title = {Dominos and the Gaussian free field},
    journal = {Ann. Probab.},
    volume = {29},
    number = {3},
    pages = {1128--1137},
    year = {2001}
}

@article{KenyonOkounkov,
author = {Richard Kenyon and Andrei Okounkov},
title = {{Limit shapes and the complex Burgers equation}},
volume = {199},
journal = {Acta Mathematica},
number = {2},
publisher = {Institut Mittag-Leffler},
pages = {263 -- 302},
year = {2007},
doi = {10.1007/s11511-007-0021-0},
URL = {https://doi.org/10.1007/s11511-007-0021-0}
}

@article{russkikh2018dimers,
  title={Dimers in piecewise Temperleyan domains},
  author={Russkikh, Marianna},
  journal={Communications in Mathematical Physics},
  volume={359},
  number={1},
  pages={189--222},
  year={2018},
  publisher={Springer}
}

@article{russkikh2020dominos,
  title={Dominos in hedgehog domains},
  author={Russkikh, Marianna},
  journal={Annales de l’Institut Henri Poincar{\'e} D},
  volume={8},
  number={1},
  pages={1--33},
  year={2020}
}

@book{gorin2021lectures,
  title={Lectures on random lozenge tilings},
  author={Gorin, Vadim},
  volume={193},
  year={2021},
  publisher={Cambridge University Press}
}

@article{nicoletti2025temperleyan,
  title={Temperleyan Domino Tilings with Holes},
  author={Nicoletti, Matthew},
  journal={arXiv preprint arXiv:2503.12082},
  year={2025}
}

@article{berggren2025gaussian,
  title={Gaussian Free Field and Discrete Gaussians in Periodic Dimer Models},
  author={Berggren, Tomas and Nicoletti, Matthew},
  journal={arXiv preprint arXiv:2502.07241},
  year={2025}
}

@article{astala2026dimer,
  title={Dimer models and conformal structures},
  author={Astala, Kari and Duse, Erik and Prause, Istv{\'a}n and Zhong, Xiao},
  journal={Communications on Pure and Applied Mathematics},
  volume={79},
  number={2},
  pages={340--446},
  year={2026},
  publisher={Wiley Online Library}
}

@article{kenyon2022gradient,
  title={Gradient variational problems in R 2},
  author={Kenyon, Richard and Prause, Istv{\'a}n},
  journal={Duke Mathematical Journal},
  volume={171},
  number={14},
  pages={3003--3022},
  year={2022},
  publisher={Duke University Press}
}

@article{basok2023dimers,
  title={Dimers on Riemann surfaces and compactified free field},
  author={Basok, Mikhail},
  journal={Annals of Probability},
  year={to appear},
  shorthand={Bas}
}

@article{BerestyckiLaslierRayI,
  title={Dimers on Riemann surfaces I: Temperleyan forests},
  author={Berestycki, Nathana{\"e}l and Laslier, Beno{\^i}t and Ray, Gourab},
  journal={Annales de l’Institut Henri Poincar{\'e} D},
  volume={12},
  number={2},
  pages={363--444},
  year={2024}
}

@article{BerestyckiLaslierRayII,
  title={Dimers on Riemann surfaces, II: Conformal invariance and scaling limit},
  author={Berestycki, Nathana{\"e}l and Laslier, Beno{\^i}t and Ray, Gourab},
  journal={Probability and Mathematical Physics},
  volume={5},
  number={4},
  pages={961--1037},
  year={2024},
  publisher={Mathematical Sciences Publishers}
}

@article{Cimasoni,
  title={Dimers on surface graphs and spin structures. I},
  author={Cimasoni, David and Reshetikhin, Nicolai},
  journal={Communications in Mathematical Physics},
  volume={275},
  pages={187--208},
  year={2007},
  publisher={Springer}
}

@article{ChelkakDeiman,
  title={Domino tilings of a black-and-white Temperleyan cylinders},
  author={Chelkak, Dmitry and Deiman, Zach},
  year={2026},
  journal={arXiv preprint arXiv:2601.13332},
  shorthand={CD}
}

@article{Quillen,
    author = {Daniel Quillen},
    title = {Determinants of Cauchy-Riemann operators over a Riemann surface},
    journal = {Funct. Anal. Its Appl.},
    volume = {19},
    pages = {31--34},
    year = {1985}
}

@book{Forster,
author="Otto Forster",
title="Lectures on Riemann Surfaces",
year="1981",
publisher="Springer New York, NY",
series="Graduate Texts in Mathematics",
edition="1"
}

@book{MumfordTata1,
author="David Mumford",
title="Tata Lectures on Theta I",
year="1983",
publisher="Birkh{\"a}user Boston, MA",
series="Modern Birkh{\"a}user Classics"
}

@book{Fay,
author="John D. Fay",
title="Theta Functions on {R}iemann Surfaces",
year="1973",
publisher="Springer, Berlin, Heidelberg",
series="Lect. Note in Math.",
volume="352"
}

@book{GriffitsHarris,
author="Phillip Griffiths and Joseph Harris",
title="Principles of Algebraic Geometry",
year="1994",
publisher="John Wiley \& Sons, Ltd"
}

\end{document}